\documentclass[12pt,final,
thmsa,
epsf,a4paper]
{article}
\usepackage{booktabs}

\newtheorem{theorem}{Theorem}

\newtheorem{corollary}{Corollary}
\newtheorem{definition}{Definition}
\newtheorem{example}{Example}
\newtheorem{lemma}{Lemma}
\newtheorem{proposition}{Proposition}
\newtheorem{remark}{Remark}

\newenvironment{proof}[1][Proof]{\emph{#1.} }{\  \hfill $\square $ \vspace{5 pt}}

\usepackage{natbib}
\usepackage{amssymb}
\usepackage[dvips]{graphics}
\usepackage{amsmath}

\usepackage{mathpazo}
\usepackage{enumerate}

\usepackage{amsfonts}

\usepackage{amssymb}

\usepackage{enumerate}

\usepackage{mathrsfs}

\usepackage[usenames]{color}
\usepackage{cancel}

\usepackage{pgf,tikz}

\usetikzlibrary{decorations.pathreplacing,angles,quotes}
\usetikzlibrary{arrows.meta}
\usetikzlibrary{arrows, decorations.markings}
\tikzset{myptr/.style={decoration={markings,mark=at position 1 with %
       {\arrow[scale=2,>=stealth]{>}}},postaction={decorate}}}

\usepackage{hyperref}
\hypersetup{
    colorlinks,
    citecolor=blue,
    filecolor=black,
    linkcolor=blue,
    urlcolor=blue
}

\usepackage{pgf,tikz}
\usetikzlibrary{arrows}

\usepackage[utf8]{inputenc}
\usepackage{amssymb, amsmath,geometry}
\usepackage{fancyhdr}

\usepackage{soul}

\usepackage{emptypage}

\newcommand*\samethanks[1][\value{footnote}]{\footnotemark[#1]}

\usepackage[T1]{fontenc}
\DeclareFontFamily{T1}{calligra}{}
\DeclareFontShape{T1}{calligra}{m}{n}{<->s*[1.44]callig15}{}
\DeclareMathAlphabet\mathcalligra   {T1}{calligra} {m} {n}

\begin{document}

\title{Cycles to compute the full set of many-to-many stable matchings
\thanks{We thank Jordi Massó for very detailed comments,   Xuan Zhang for the counterexample in Example \ref{ejemplo 2} and two anonymous referees. We acknowledge financial support
from UNSL through grants 032016 and 030320, from Consejo Nacional
de Investigaciones Cient\'{\i}ficas y T\'{e}cnicas (CONICET) through grant
PIP 112-200801-00655, and from Agencia Nacional de Promoción Cient\'ifica y Tecnológica through grant PICT 2017-2355.}}


\author{Agustín G. Bonifacio\thanks{Instituto de Matem\'{a}tica Aplicada San Luis, Universidad Nacional de San
Luis and CONICET, San Luis, Argentina, and RedNIE. Emails: \texttt{abonifacio@unsl.edu.ar} (A. G. Bonifacio), 
\texttt{nmjuarez@unsl.edu.ar} (N. Juarez), \texttt{paneme@unsl.edu.ar} (P. Neme) and \texttt{joviedo@unsl.edu.ar} (J. Oviedo).} \and  
Noelia Juarez\samethanks[2] \and Pablo Neme\samethanks[2] \and Jorge Oviedo\samethanks[2]}

\date{\today}
\maketitle

\begin{abstract}
In a many-to-many matching model in which agents' preferences satisfy substitutability and the law of aggregate demand, we present an algorithm to compute the full set of stable matchings. This algorithm relies on the idea of ``cycles in preferences''  and  generalizes the algorithm presented in \cite{roth1992two} for the one-to-one model.  

\bigskip

\bigskip


\end{abstract}

\section{Introduction}

In many-to-many matching models, there are two disjoints sets of agents: firms and workers. Each firm wishes to hire a set of workers and each worker wishes to work for a set of firms. Many real-world markets are many-to-many, for instance, the market for medical interns in the UK \citep{roth1992two}, the assignment of teachers to high schools in some countries (35\% of teachers in Argentina work in more than one school). A matching is an assignment of sets of workers to firms, and of sets of firms to workers, so that a firm is assigned to a worker if and only if this worker is also assigned to that firm. In these models, the most studied  solution is the set of stable matchings. 
A matching is stable if all agents are matched to an  acceptable subset of partners and there is no unmatched firm-worker pair, both of which would prefer to add the other to their current subset of partners.\footnote{This notion of stability is known in the literature as ``pairwise stability''.} In their seminal paper, \cite{gale1962college} introduce the Deferred Acceptance (\emph{DA}, from now on) algorithm to show the existence of a stable matching in the one-to-one model.  This algorithm computes the optimal stable matching for one side of the market. Later, the \emph{DA} algorithm  is adapted to the many-to-many case by \cite{roth1984evolution}. 

In this paper, we present an algorithm to compute the full set of many-to-many stable matchings. In the one-to-one model, beginning from a stable matching and through a procedure of reduction of preferences, \cite{roth1992two} define a ``cycle in preferences'' that allows them to generate a new  matching, called a ``cyclic matching'', that turns out to be stable.\footnote{\cite{roth1992two} adapt the algorithm presented in \cite{irving1986complexity}. Cycles are called ``rotations'' in \cite{irving1986complexity} } They present an algorithm that, starting from an optimal stable matching for one side of the market and by constructing all cycles and its corresponding cyclic matchings, computes the full set of one-to-one stable matchings \citep[see][for more details]{irving1986complexity,gusfield1989stable,roth1992two}.  
The purpose of our paper is to extend Roth and Sotomayor's construction to a many-to-many environment.

Our general framework assumes substitutability on all agents' preferences. This condition, first introduced by \cite{kelso1982job}, is the weakest requirement in preferences in order to guarantee the existence of many-to-many stable matchings. An agent has substitutable preferences when she wants to continue being matched to an agent of the other side of the market even if other agents become unavailable.  Given an agent's preference, \cite{blair1988lattice} defines a partial order over subsets of agents of the other side of the market as follows: one subset is  Blair-preferred to another subset if, when all agents of both subsets are available, only the agents of the first subset are chosen.\footnote{Blair's order of an agent is more restrictive than the individual preference of that agent.}
When preferences are substitutable, the set of stable matchings has a lattice structure with respect to the unanimous Blair order for any side of the market.\footnote{For instance, a set of workers is Blair-preferred to another set of workers for the firms if the first set is Blair-preferred to the latter set for each firm.}

In addition to substitutability, we require that agents' preferences satisfy the ``law of aggregate demand" (\emph{LAD}, from now on).\footnote{This property is first studied by \cite{alkan2002class} under the name of ``cardinal monotonicity". See also \cite{hatfield2005matching}.}  This condition says that when an agent chooses from an expanded set, it selects at least as many agents as before. Under these two assumptions on preferences, the set of stable matchings   satisfies the so-called Rural Hospitals Theorem, which states that each agent is matched with the same number of partners in every stable matching. 
Substitutability of preferences and \emph{LAD} ensure that  suitable generalizations of the concepts of ``cycle'' and ``cyclic matching'' can be defined.   To do this, given a substitutable preference profile and two stable matchings that are unanimously Blair-comparable (for one side of the market), we define a ``reduced preference profile'' with respect to these two stable matchings and show that this profile is also substitutable and satisfies \emph{LAD}. 
Next, we adapt Roth and Sotomayor's notion of a cycle for our reduced preference profile and  use this many-to-many notion of a cycle to define a cyclic matching. This new matching turns out to be stable not only for this reduced preference profile but also for the original preference profile. 
With all these ingredients we can describe our algorithm as follows. Given a preference profile, by the \emph{DA} algorithm compute the two optimal stable matchings, one for each side of the market. Pick one side of the market, say the firms' side, and  obtain the reduced preference profile with respect to the firms' optimal and the workers' optimal stable matchings. In each of the following steps, for each  reduced preference profile obtained in the previous step,  compute: (i) each cycle for this profile, (ii) its corresponding cyclic matching, and (iii) the reduced preference profile with respect to this cyclic matching and the worker optimal stable matching. The algorithm stops in the step where all the cyclic matchings computed are equal to the worker optimal stable matching. The firms' optimal stable matching together with all the cyclic matchings obtained by the algorithm encompass the full set of stable matchings.

Several papers calculate the full set of stable matchings in two-sided matching models.  \cite{mcvitie1971stable} are the first to present an algorithm that computes the full set of one-to-one stable matchings. This algorithm starts at the optimal stable matching for one side of the market and then, at each step, breaks some matched pair and applies the \emph{DA} algorithm to the new preference profile in which the broken matched pair is no longer acceptable.  This algorithm is generalized by \cite{martinez2004algorithm} to a many-to-many matching market in which agents' preferences satisfy substitutability. 
However, we provide an example that shows that the algorithm in \cite{martinez2004algorithm} has an error: it stops before computing all stable matchings. We also give an intuition of why this happens.

Following the lines of \cite{irving1986complexity} and \cite{roth1992two}, \cite{bansal2007polynomial} extend the notion of cycle to a many-to-many matching model in which each agent has a strict ordering over individual agents of the other side of the market.\footnote{This setting is equivalent to the one defined by \cite{roth1985college} for the many-to-many model in which firms have responsive preferences over subsets of workers.} Among other results, they use cycles to compute the full set of stable matchings. \cite{eirinakis2012finding} revise and improve the algorithm presented in \cite{bansal2007polynomial}. Moreover, they extend the algorithm for a model in which agents' preferences satisfy the ``max-min criteria''. This criteria establishes that agents rank stable matchings in a responsive manner. However, their assumptions are more restrictive than substitutability over subsets of agents and \emph{LAD}.  For a many-to-one matching model with strict orderings over individual agents \cite{cheng2008unified}, using the notion of cycles introduced by \cite{bansal2007polynomial}, 
show that broad classes of feasibility and optimization stable matching problems can be solved efficiently.

A different approach to compute the full set of stable matchings is presented by \cite{dworczak2021deferred}. For a one-to-one model,  they generalize the \emph{DA} algorithm allowing both sides of the market to make offers in a specific ordering. The paper proposes a generalized \emph{DA} algorithm with ``compensation chains'' and proves that: (i) for each order of the agents, the algorithm obtains a stable matching, 
and (ii) each stable matching can be obtained as the output of the algorithm for some order of the agents.

Our paper is organized as follows. In Section  \ref{seccion de preliminares} we present the preliminaries. The reduction procedure of preferences is presented in Section  \ref{seccion de resuccion de preferencias}. Section \ref{seccion de cyclos y algoritmo} contains the definition of a cycle in preferences together with the algorithm that computes the many-to-many stable set. Concluding remarks are gathered in Section \ref{Concludings}, where the error in \cite{martinez2004algorithm} is discussed. All proofs are relegated to Appendix \ref{appendix}.

\section{Preliminaries}\label{seccion de preliminares}
We consider many-to-many matching markets where there are two disjoint sets of agents: the set of firms $F$ and the set of workers $W$. Each firm $f\in F$ has a strict preference relation $P_f$ over the  set of all subsets of $W$.  Each worker $w\in W$ has a strict preference relation $P_w$ over the  set of all subsets of $F$. 
We denote by $P$ the preference profile for all agents: firms and workers.  A (many-to-many) matching market is denoted by $(F,W,P).$ Since the sets $F$ and $W$ are kept fixed throughout the paper, we often identify the market $(F,W,P)$ with the preference profile $P$.  Given an agent $a\in F\cup W$, a set $S$ in the opposite side of the market is \textbf{acceptable for $\boldsymbol{a$ under $P}$} if $S P_a \emptyset$. A pair $(f,w)\in F\times W$ is \textbf{mutually acceptable under $\boldsymbol{P}$} if $\{f\}$ is acceptable for $w$ under $P$ and $\{w\}$ is acceptable for $f$ under $P$. In this paper, the preference relation $P_a$ is represented by the ordered list of its acceptable sets (from most to least preferred).\footnote{For instance, $ P_{f}:w_1w_2,w_3,w_1,w_2 $ indicate that $ \{ w_1,w_2\} P_f \{ w_3\} P_f\{ w_1\} P_f \{ w_2\} P_f \emptyset $ and $ P_{w}:f_1f_3,f_3,f_1 $ indicates that $ \{f_1,f_3 \} P_w \{ f_3 \} P_w \{ f_1 \} P_w \emptyset$.}
Given a set of workers $W'\subseteq W$ and a firm $f\in F$, let $\boldsymbol{C_f(W')}$ (the choice set for $f$) denote firm $f$'s most preferred subset of $W'$ according to the preference relation $P_f$. 
Symmetrically, given a set of firms $F'\subseteq F$ and a worker $w\in W$, let $\boldsymbol{C_w(F')}$ (the choice set for $w$) denote worker $w$'s most preferred subset of $F'$ according to the preference relation $P_w$.

\begin{definition}
A \textbf{matching} $\mu$ is a function from the set $F\cup W$ into $2^{F\cup W}$ such that for each $w\in W$ and for each $f\in F$:
\begin{enumerate}[(i)]
\item $\mu(w)\subseteq F$,
\item $\mu(f)\subseteq W$,
\item $w\in \mu(f)$ if and only if $f\in \mu(w)$.
\end{enumerate}
\end{definition}
Agent $a\in F\cup W$ is \textbf{matched} if $\mu(a) \neq \emptyset$, otherwise she is \textbf{unmatched}. For the following definitions, fix a preference profile $P$.
A matching $\mu$ is \textbf{blocked by agent $\boldsymbol{a}$} if $\mu(a)\neq C_a(\mu(a))$. A matching is  \textbf{individually rational} if it is not blocked by any individual agent. A matching $\mu$ is \textbf{blocked by a firm-worker pair $\boldsymbol{(f,w)}$} if $w \notin \mu( f ), w \in C_f(\mu( f )\cup \{w\}),$ and $f \in C_w(\mu( w )\cup \{f\})$. A matching $\mu$ is \textbf{stable} if it is not blocked by any individual agent or any firm-worker pair. The set of stable matchings  for a preference profile $P$ is denoted by $\boldsymbol{S(P)}.$ 

Agent $a$'s preference relation satisfies \textbf{substitutability} if, for each subset $S$ of the opposite side of the market (for instance, if $a\in F$ then $S\subseteq W$) that contains agent $b$, $b\in C_a(S)$ implies that $b\in C_a(S'\cup \{b\})$ for each $S' \subseteq S.$ Moreover, if agent $a$'s preference relation is substitutable then it holds that
\begin{equation}\label{propiedad 1}
C_a(S\cup S')=C_a(C_a(S)\cup S')
\end{equation} for each pair of subsets $S$ and $S'$ of the opposite side of the market.\footnote{See Proposition 2.3 in \cite{blair1988lattice}, for more details.}  

 Given a firm $f$, \cite{blair1988lattice} defines a partial order for $f$  over  subsets of workers  as follows: given  firm $f$'s preference relation $P_f$ and two subsets of workers  $S$ and  $S'$, we write $\boldsymbol{S \succeq_f S'}$ whenever $S=C_f(S \cup S')$, and $\boldsymbol{S \succ_f S'}$ whenever $S \succeq_f S'$ and $S \neq S'$. The partial orders $\succeq_w$ and $\succ_w$ for  worker $w$ are defined analogously. Given a preference profile $P$ and two  matchings $\mu$ and  $\mu'$,  we write  $\boldsymbol{\mu \succeq_F \mu'}$ whenever $\mu(f) \succeq_f \mu'(f)$ for each $f\in F$, and we write $\boldsymbol{\mu \succ_F \mu'}$ if, in addition, $\mu \neq \mu'$.\footnote{We call $\succeq_F$ the unanimous Blair order for the firms.}  Similarly, we define $\succeq_W$ and $\succ_W$. 

The set of stable matchings under substitutable preferences is very well structured. \cite{blair1988lattice} proves that this set has  two lattice structures, one with respect to $\succeq_F$ and the other one with respect to $\succeq_W$. Furthermore, it contains two distinctive matchings: the firm-optimal stable matching $\mu_F$ and the worker-optimal stable matching $\mu_W$. The  matching $\mu_F$ is unanimously considered by all firms to be the best among all stable matchings and $\mu_W$ is unanimously considered by all workers to be the best among all stable matchings,  according to the respective Blair's partial orders \cite[see][for more details]{roth1984evolution,blair1988lattice}.

Agent $a$'s preference relation   satisfies the \textbf{ law of aggregate demand (\emph{LAD})} if for all subsets $S$ of the opposite side of the market and all $S'\subseteq S$, $|C_a(S')|\leq |C_a(S)|.$\footnote{$|S|$ denotes the number of agents in $S$.}  When preferences are substitutable and satisfy \emph{LAD},  the  lattices $(S(P),\succeq_F)$ and $(S(P),\succeq_W)$ are dual; that is, $\mu\succeq_F\mu'$ if and only if  $\mu'\succeq_W\mu$ for $\mu,\mu'\in S(P)$. This is known as the ``polarization of interests'' result \citep[see][among others]{alkan2002class,li2014new}.

\section{The reduction procedure}\label{seccion de resuccion de preferencias}
In this section, we present  a reduction procedure that will allow us to define a cycle in preferences, a concept  that is essential  for developing our algorithm. Given a substitutable preference profile and two Blair-comparable (for the firms) stable matchings, this reduction procedure generates a new preference profile, in which the most Blair-preferred stable matching is the firm-optimal matching and the least Blair-preferred stable matching is the worker-optimal matching, for the market identified with this new preference profile. The reduction procedure is described as follows.
Let $\mu $ and $\widetilde{\mu}$ be stable matchings for matching market $(F,W,P)$ such that $\mu \succeq_F \widetilde{\mu}$. 

\begin{enumerate}
\item[\textbf{Step 1:}] 

\begin{enumerate}
[(a)]

\item For each $f\in F$, each $W^{\prime }\subset W$ such that $W^{\prime
}\succ_f \mu (f),$  and each $\tilde{w}\in
W^{\prime }\setminus \mu (f)$, remove each $\widetilde{W}\subset W$ such
that $\tilde{w}\in \widetilde{W}$ from $f$'s list of acceptable sets of
workers.

\item For each $w\in W$, each $F^{\prime }\subset F$ such that $F^{\prime
}\succ_w \widetilde{\mu}(w),$ and each $\tilde{f}\in
F^{\prime }\setminus \widetilde{\mu}(w)$, remove each $\widetilde{F}\subset F$ such
that $\tilde{f}\in \widetilde{F}$ from $w$'s list of acceptable sets of
firms. 
\end{enumerate}

\item[\textbf{Step 2:}] 

\begin{enumerate}[(a)]

\item For each $f\in F$, each $W^{\prime }\subset W$ such that $\widetilde{\mu}(f)\succ_f W',$ and each $\tilde{w}%
\in W^{\prime }\setminus \widetilde{\mu}(f)$, remove each $\widetilde{W}\subset W$
such that $\tilde{w}\in \widetilde{W}$ from $f$'s list of acceptable sets of
workers.

\item For each $w\in W$, each $F^{\prime }\subset F$ such that $\mu
(w)\succ_w F^{\prime },$ and each $\tilde{f}\in
F^{\prime }\setminus \mu (w)$, remove each $\widetilde{F}\subset F$ such
that $\tilde{f}\in \widetilde{F}$ from $w$'s list of acceptable sets of
firms.
\end{enumerate}

\item[\textbf{Step 3:}] After Steps 1 and 2 are performed, if  $f$ is not acceptable for $w$ (that is, if $\{f\}$ is not on $w$'s preference list as now modified), remove each $%
W^{\prime }\subset W$ such that $w\in W^{\prime }$ from $f$'s list of
acceptable sets of workers. If  $w$
is not acceptable for $f$ (that is, if $\{w\}$ is not on $f$'s preference list as now modified), remove each $F^{\prime }\subset F$ such that $%
f\in F^{\prime }$ from $w$'s list of acceptable sets of firms.
\end{enumerate}

\noindent The  profile obtained by this procedure  is called the \textbf{reduced preference profile with respect to $\boldsymbol{\mu$ and $\widetilde{\mu}}$}, and is denoted by $\boldsymbol{P^{\mu,\widetilde{\mu}}}$. When $\widetilde{\mu}=\mu_W,$ the  profile is simply called the \textbf{reduced preference profile with respect to $\boldsymbol{\mu}$}, and is denoted by $\boldsymbol{P^{\mu}}$.

Let us put in words how the reduction procedure works. In Step 1 (a), for each $f\in F$, if a worker is not in $\mu(f)$ but belongs to  a subset that is Blair-preferred to $\mu(f)$, the procedure eliminates each subset that contains this worker from firm $f$'s list of acceptable subsets. Step 1 (b) performs an analogous elimination in each worker's preference list. In Step 2 (a), for each $f\in F$,  if a worker is not in $\widetilde{\mu}(f)$ and $\widetilde{\mu}(f)$ is Blair-preferred to a subset that includes this worker, the procedure eliminates each subset that contains this worker from firm $f$'s list of acceptable subsets. Step 2 (b) performs an analogous elimination in each worker's preference list. In Step 3, after Step 1 and Step 2  are performed, the procedure eliminates all subsets of agents needed in order to make all pairs of agents mutually acceptable.

 By  $\boldsymbol{C^{\mu,\widetilde{\mu}}_f(W')}$ we denote the firm $f$'s most preferred subset of $W'$ according to the preference relation $P^{\mu,\widetilde{\mu}}_f$. Similar notation is used for the choice sets according to the preference relations $P^{\mu,\widetilde{\mu}}_w$, $P^{\mu}_f$, and $P^{\mu}_w$. Some remarks on the reduced preference relations
are in order.

\begin{remark}
\label{remark de reduccion M-M} 
Let $P$ be a market and assume $\mu,\widetilde{\mu}\in S(P)$. Then the following statements hold.
\begin{enumerate}[(i)]

\item $\mu \left( f\right) $ is the most preferred subset of workers in $f$%
's reduced preference relation (i.e. $\mu \left( f\right) =C_{f}^{\mu,\widetilde{\mu} }(W)$) and $%
\widetilde{\mu}(w)$ is the most preferred subset of firms in $w$'s reduced preference relation (i.e.  $\widetilde{\mu}\left( w\right) =C_{w}^{\mu,\widetilde{\mu} }(F)$). 

\item $\mu $ is the firm--optimal stable matching under $P^{\mu,\widetilde{\mu} }$ and $%
\widetilde{\mu}$ is the worker--optimal stable matching under $P^{\mu,\widetilde{\mu} }$.
Furthermore, $\widetilde{\mu}$ is the firm--pessimal stable matching under $P^{\mu,\widetilde{\mu} }$ and $\mu $ is the worker--pessimal stable matching under $P^{\mu,\widetilde{\mu} }
$.

\item $f$ is acceptable to $w$ if and only if $w$ is acceptable to $f$ under 
$P^{\mu,\widetilde{\mu} }.$ 



\end{enumerate}
\end{remark}

The following lemma states that the properties
of substitutability and \emph{LAD} are preserved by the reduction procedure.

\begin{lemma}\label{preferenceias resucidas sustituibles}
Let $\mu,\widetilde{\mu}\in S(P)$ and  $a \in F\cup W$. If $P_a$ is substitutable and satisfies LAD, then the reduced preference relation $P^{\mu,\widetilde{\mu}}_{a}$ is substitutable and satisfies LAD.
\end{lemma}

The following example illustrates the reduction procedure for a matching market.

\begin{example}\label{Ejemplo 1}
Let $(F,W,P)$ be a matching market where $F=\{f_1,f_2,f_3\}$, $W=\{w_1,w_2,w_3,w_4,$ $w_5,w_6\}$, and  the preference profile is given by:
\medskip

\noindent\begin{tabular}{l}
$P_{f_1}:w_1w_2,w_1w_5,w_2w_5,w_1w_3,w_4w_5,w_2w_4,w_1w_4,w_3w_4,w_3w_5,
w_2w_3,w_1,w_4,w_3,w_2,w_5$\\
$P_{f_2}:w_3w_6,w_3w_5,w_5w_6,w_2w_5,w_1w_3,w_2w_6,w_1w_5,w_1w_2,w_2w_3,w_1w_6,w_1,w_2,w_3,w_5,w_6$\\
$P_{f_3}:w_2w_4,w_1w_2,w_3w_4,w_2w_3,w_1w_3,w_1w_4,w_1,w_2,w_3,w_4$\\
$ P_{w_1}:f_3,f_1,f_2$\\
$ P_{w_2}:f_2f_3,f_1f_3,f_1f_2,f_1,f_2,f_3$\\
$ P_{w_3}:f_1,f_2$\\
$ P_{w_4}:f_1, f_3,f_2$\\
$ P_{w_5}:f_2,f_3$\\
$ P_{w_6}:f_1f_3,f_3,f_1 $\\
\end{tabular}\medskip

\noindent It is easy to check that these preference relations are substitutable and satisfy LAD. By the DA algorithm, we obtain the two optimal stable matchings: \medskip

$
\mu_F=\begin{pmatrix}

f_1 & f_2 & f_3 & \emptyset  \\
 w_1w_2& w_3w_5 &  w_2w_4& w_6 \\
 
\end{pmatrix}
$
 and $
\mu_W=\begin{pmatrix}

f_1 & f_2 & f_3 & \emptyset \\
 w_3w_4& w_2w_5 &  w_1w_2& w_6 \\
 
\end{pmatrix}.$
\bigskip

\noindent \medskip Now, after the reduction procedure is performed, we obtain the reduced preference profile with respect to $\mu_F$, $P^{\mu_F}$:\footnote{Notice that the subsets assigned in the optimal stable matchings are in bold.}
 
 \noindent\begin{tabular}{l}
$P^{\mu_F}_{f_1}:\boldsymbol{w_1w_2},w_1w_3,w_2w_4,w_1w_4,\boldsymbol{w_3w_4},w_2w_3,w_1,w_4,w_3,w_2$\\
$P^{\mu_F}_{f_2}:\boldsymbol{w_3w_5},\boldsymbol{w_2w_5},w_2w_3,w_2,w_3,w_5$\\
$P^{\mu_F}_{f_3}:\boldsymbol{w_2w_4},\boldsymbol{w_1w_2},w_1w_4,w_1,w_2,w_3,w_4$\\
 $P^{\mu_F}_{w_1}:\boldsymbol{f_3},\boldsymbol{ f_1}$\\
$P^{\mu_F}_{w_2}:\boldsymbol{f_2f_3},\boldsymbol{f_1f_3},f_1f_2,f_1,f_2,f_3$\\
$P^{\mu_F}_{w_3}:\boldsymbol{f_1},\boldsymbol{f_2}$\\
$P^{\mu_F}_{w_4}:\boldsymbol{f_1},\boldsymbol{f_3}$\\
$P^{\mu_F}_{w_5}:\boldsymbol{f_2}$\\
$P^{\mu_F}_{w_6}:\boldsymbol{\emptyset}$\\
\end{tabular}\medskip
 
\noindent In order to show how each stage of the procedure works, we turn our attention to  preferences  $ P_{f_1}  $ and $ P_{f_2} $. At Step 1 of the reduction procedure we remove the following subsets of agents: \medskip
 
 \noindent\begin{tabular}{l}
$P_{f_1}:\boldsymbol{w_1w_2},w_1w_5,w_2w_5,w_1w_3,w_4w_5,w_2w_4,w_1w_4,\boldsymbol{w_3w_4},w_3w_5,
w_2w_3,w_1,w_4,w_3,w_2,w_5$\\
$P_{f_2}:\cancel{w_3w_6},\boldsymbol{w_3w_5},\cancel{w_5w_6},\boldsymbol{w_2w_5},w_1w_3,\cancel{w_2w_6},w_1w_5,w_1w_2,w_2w_3,\cancel{w_1w_6},w_1,w_2,w_3,w_5,\cancel{w_6}.$\\
\end{tabular}\bigskip

\noindent At Step 2 of the reduction procedure we remove the following subsets of agents: \bigskip

\noindent\begin{tabular}{l}
$P_{f_1}:\boldsymbol{w_1w_2},w_1w_5,w_2w_5,w_1w_3,w_4w_5,w_2w_4,w_1w_4,\boldsymbol{w_3w_4},w_3w_5,
w_2w_3,w_1,w_4,w_3,w_2,w_5$\\
$P_{f_2}:\boldsymbol{w_3w_5},\boldsymbol{w_2w_5},\bcancel{w_1w_3},\bcancel{w_1w_5},\bcancel{w_1w_2},w_2w_3,\bcancel{w_1},w_2,w_3,w_5.$\\
\end{tabular}\bigskip

\noindent Since $ f_1 $ is not acceptable for $ w_5 $ at the original preferences, $ f_1 $ is not acceptable for $ w_5 $ after Steps 1 and 2 are performed. So at Step 3 we remove the following subsets of agents:\bigskip

\noindent\begin{tabular}{l}
$P_{f_1}:\boldsymbol{w_1w_2},\xcancel{w_1w_5},\xcancel{w_2w_5},w_1w_3,\xcancel{w_4w_5},w_2w_4,w_1w_4,\boldsymbol{w_3w_4},\xcancel{w_3w_5},
w_2w_3,w_1,w_4,w_3,w_2,\xcancel{w_5}$\\
$P_{f_2}:\boldsymbol{w_3w_5},\boldsymbol{w_2w_5},w_2w_3,w_2,w_3,w_5.$\\
\end{tabular}\bigskip

\noindent In this way we obtain $ P^{\mu_F}_{f_1}  $ and $ P^{\mu_F}_{f_2} $.
 \hfill $\Diamond$
\end{example}

\bigskip


The following theorem states that the stability of a matching is preserved by the reduction procedure and that there are no new stable matchings for the reduced preference profile. This means that  a stable matching in the original preference profile is in between (according to Blair's partial order) of the two stable matchings used to generate the reduced preference profile if and only if it is also stable in the reduced preference profile.\footnote{Recall that $\succeq_F$ and $\succeq_W$ are dual orders only in the set of stable matchings.} An important fact about this theorem (and its corollary) is that \emph{LAD} is not needed to obtain it.\footnote{In this paper there are only three results in which \emph{LAD} is not needed: Theorem \ref{estable original sii estable en el reducido M-to-M}, Corollary \ref{corolario 1}, and Lemma \ref{individual dentro de mu y mu tilde}. }

\begin{theorem}\label{estable original sii estable en el reducido M-to-M}
Let $\mu,\widetilde{\mu}\in S(P)$ with $\mu\succeq_{F}\widetilde{\mu}$. Then, $\mu'\in S(P)$ and  $\mu\succeq_{F} \mu' \succeq_{F}\widetilde{\mu}$ if and only if  $\mu' \in S(P^{\mu,\widetilde{\mu}}).$
\end{theorem}

Notice that by optimality of $\mu_F$ and $\mu_W$, any stable matching $\mu\in S(P)$ satisfies $\mu_F\succeq_{F} \mu$ and $\mu_W\succeq_W \mu$. Furthermore, by the polarization of interests, $\mu \succeq_F \mu_W.$ Then, $\mu_F \succeq_{F} \mu \succeq_F \mu_W$. Thus, as a consequence of  Theorem \ref{estable original sii estable en el reducido M-to-M} we can state the following corollary.

\begin{corollary}\label{corolario 1}
$S(P)=S(P^{\mu_{F}})$.
\end{corollary}

\section{Cycles and Algorithm}\label{seccion de cyclos y algoritmo}

In this section, we present the algorithm to compute the full set of many-to-many stable matchings. First, we introduce  its key ingredients: the notion of a cycle in preferences and its corresponding cyclic matching. From now on,   we assume that the preferences of all agents are substitutable and satisfy \emph{LAD}.

\subsection{Cycles and cyclic matchings}\label{subseccion de ciclos y matching ciclicos}

In the one-to-one model, \cite{roth1992two} present  the notion of a cycle in preferences.\footnote{\cite{roth1992two} adapt  the notion of \emph{rotation} presented in  \cite{irving1986complexity}, and refer to it as a cycle in preferences.} Their construction can be roughly explained as follows. Consider a one-to-one matching market  $(M,W,P)$ and a stable matching $\mu \in S(P).$  A reduced preference profile with respect to $\mu$ and the worker-optimal stable matching $\mu_W,$ say $P^{\mu, \mu_W},$ is obtained. The important facts about this reduced profile are that: (i) $\mu(f)$ is $f$'s most preferred partner and $\mu_W(f)$  is $f$'s least preferred partner according to $P^{\mu, \mu_W}_f$, for each $f \in F$; and (ii)    $\mu_W(w)$ is $w$'s most preferred partner and $\mu(w)$  is $w$'s least preferred partner according to $P^{\mu, \mu_W}_w$, for each $w \in W.$    A  cycle for $P^{\mu, \mu_W}$ in the one-to-one model can be seen as an ordered  sequence of worker-firm pairs $\{(w_{1},f_{1}),(w_{2},f_{2}),\ldots,(w_r ,f_{r})\}$ such that $w_i=\mu(f_i)$ 
 and  $w_{i+1}$ is $f_i$'s  most-preferred worker of $W \setminus \{w_i\}$ according to $P^{\mu, \mu_W}_{f_i}.$ Our definition of a cycle generalizes this idea to the many-to-many environment. Formally,

\begin{definition}\label{defino cyclo}
Let $\mu,\widetilde{\mu}\in S(P)$ with $\mu \succ_F \widetilde{\mu}.$ A \textbf{cycle $\boldsymbol{\sigma$ for $P^{\mu,\widetilde{\mu}}}$} is an ordered  sequence of worker-firm pairs $\sigma=\{(w_{1},f_{1}),(w_{2},f_{2}),\ldots,(w_r ,f_{r})\}$  such that, for $i=1,\ldots,r,$ we have:
\begin{enumerate}[(i)]
\item $w_i\in \mu(f_i)\setminus \widetilde{\mu}(f_{i})$, 
\item $C^{\mu,\widetilde{\mu}}_{f_i}(W \setminus\{w_{i}\}) =\left(\mu( f_i) \setminus \{w_i\}\right)\cup \{w_{i+1}\}$, with $w_{r+1}=w_{1},$ and
\item  $C^{\mu,\widetilde{\mu}}_{w_i}(\mu(w_i) \cup \{f_{i-1}\}) =\left(\mu(w_i) \setminus \{f_i\}\right)\cup \{f_{i-1}\}$, with $f_{0}=f_{r}.$
\end{enumerate} 
\end{definition}
Condition (i) states that worker $w_i$ is matched with $f_i$ under $\mu$ but not under $\widetilde{\mu}.$ Condition (ii) states that the set obtained from $\mu(f_i)$ by eliminating worker $w_i$ and adding worker $w_{i+1}$ is the most Blair-preferred subset of workers of $W\setminus \{w_i\}$ that  contains $w_{i+1}$, according to $P^{\mu, \widetilde{\mu}}_{f_i}.$ Condition (iii) mimics Condition (ii) for  the other side of the market: it states that the set obtained from $\mu(w_i)$ by eliminating firm $f_i$ and adding firm $f_{i-1}$ is the least Blair-preferred subset of firms among those that are Blair-preferred to $\mu(w)$  and contains $f_{i-1}$, according to $P^{\mu, \widetilde{\mu}}_{w_i}.$ Notice that Condition (iii) is not needed in the one-to-one model.

In the rest of this section, we state four propositions that are essential to show that the algorithm computes the full set of stable matchings. All the proofs are relegated to the appendix.
The following proposition  gives a necessary and sufficient condition for the existence of a cycle in a reduced preference profile.
\begin{proposition}\label{existencia de ciclo} 
Let $\mu, \widetilde{\mu} \in S(P )$ with $\mu\succeq_{F}\widetilde{\mu}$. There is a cycle  for  $P^{\mu,\widetilde{\mu}}$ if and only if $%
\mu \neq \widetilde{\mu}$.
\end{proposition}

In the one-to-one model, a cycle $\{(w_{1},f_{1}),(w_{2},f_{2}),\ldots,(w_r ,f_{r})\}$ for $P^{\mu, \mu_W}$ can be used to obtain a new matching from matching $\mu$ by breaking the partnership between firm $f_i$ and worker $w_i$ and establishing a new partnership between firm $f_i$ and worker $w_{i+1}$ for each $i=1, \ldots, r$ (modulo $r$), keeping all remaining partnerships in $\mu$ unaffected.  This new matching is called a cyclic matching. Using our many-to-many version of a cycle, we generalize the concept of cyclic matching in a straightforward way:

\begin{definition}\label{defino matching ciclico}
Let $\mu,\widetilde{\mu } \in S(P)$ with $\mu \succ_F \widetilde{\mu},$ and let $\sigma=\{(w_{1},f_{1}),(w_{2},f_{2}),\ldots,(w_r ,f_{r})\} $ be a cycle for $P^{\mu ,\widetilde{\mu }}$. The \textbf{cyclic matching $\boldsymbol{\mu _{\sigma }$ under $P^{\mu ,\widetilde{\mu }}}$} is defined as follows: for each $f\in F$

\[
\mu _{\sigma }\left( f\right) =\left\{ 
\begin{array}{lcl}
 \Big[ \mu (f) \setminus  \{w_i \ : \ f=f_i\} \Big]   \bigcup \{w_{i+1} \ : \ f=f_i\} &  & \text{if } f \in \sigma \\ 
&  &  \\ 
\mu (f) &  & \text{if } f\notin \sigma,
\end{array}%
\right. 
\]%

\noindent and for each $ w \in W$, $\mu_{\sigma}(w)=\{ f\in F: w\in \mu_{\sigma}(f) \} .$
\end{definition}

For Example \ref{Ejemplo 1}, we illustrate how to compute a cycle and its corresponding cyclic matching.

\noindent \textbf{Example 1 (Continued)} \textit{ $\sigma_{1}=\lbrace (w_1,f_1),(w_4,f_3) \rbrace $ is a cycle for $ P^{\mu_F} $ in Example \ref{Ejemplo 1}.
 To see this, we show that each worker-firm pair in  $\sigma_{1} $ satisfies (i), (ii) and (iii) of Definition \ref{defino cyclo}. }
  \begin{enumerate}
\item[\textit{(i)}] \textit{By inspection,} $ w_1\in \mu_F(f_1)\setminus \mu_W(f_1) $ \textit{and}  $ w_3\in \mu_F(f_3)\setminus \mu_W(f_3) $.
\item[\textit{(ii)}] $ C_{f_1}(W \setminus  \lbrace w_1   \rbrace )=C^{\mu_F}_{f_1}( \lbrace w_2,w_3,w_4,w_5,w_6 \rbrace )= \lbrace w_2,w_4 \rbrace = \left(\mu_F(f_1) \setminus  \lbrace w_1  \rbrace\right) \cup \lbrace w_2 \rbrace$,

 $ C_{f_3}(W \setminus  \lbrace w_4   \rbrace )=C^{\mu_F}_{f_3}( \lbrace w_1,w_2,w_3,w_5,w_6 \rbrace )= \lbrace w_1,w_2 \rbrace = \left(\mu_F(f_3) \setminus  \lbrace w_4  \rbrace\right) \cup \lbrace w_1 \rbrace$.
 \item[\textit{(iii)}] $ C^{\mu_F}_{w_1}(\mu_F(w_1) \cup \lbrace f_3 \rbrace)=  C^{\mu_F}_{w_1}(\lbrace f_1,f_3 \rbrace )= \lbrace f_3 \rbrace = \left(\mu_F(w_1) \setminus \lbrace f_1 \rbrace\right) \cup \lbrace f_3 \rbrace ,$
 
 $ C^{\mu_F}_{w_4}(\mu_F(w_4) \cup \lbrace f_1 \rbrace)=  C^{\mu_F}_{w_4}(\lbrace f_3,f_1 \rbrace )= \lbrace f_1 \rbrace = \left(\mu_F(w_4) \setminus \lbrace f_3 \rbrace \right)\cup \lbrace f_1 \rbrace .$

  \end{enumerate}
 
\noindent  \textit{Now, we compute its associated cyclic  matching $ \mu_{\sigma_1} $. Since $ f_1 $ and $ f_3 $ are firms in  $ \sigma_1 $, then
   $ \mu_{\sigma_1}(f_1)= \left(\mu_F(f_1) \setminus  \lbrace w_1  \rbrace \right)\cup \lbrace w_2 \rbrace = \lbrace w_2,w_4 \rbrace $ and $ \mu_{\sigma_1}(f_3)= \left(\mu_F(f_3) \setminus  \lbrace w_4  \rbrace \right)\cup \lbrace w_1 \rbrace = \lbrace w_1,w_2 \rbrace $. Thus,}
$
\mu_{\sigma_1}= \begin{pmatrix}

f_1 & f_2 & f_3 & \emptyset  \\
 w_2w_4& w_3w_5 &  w_1w_2& w_6 \\
 
\end{pmatrix}
$. \hfill $\Diamond$
\bigskip

In the next proposition, we state that each cyclic matching under a reduced preference profile is stable for that same reduced preference profile.
\begin{proposition}\label{ciclico es estable}
Let $\mu, \widetilde{\mu} \in S(P )$ with $\mu\succ_{F}\widetilde{\mu}$ and let  $\mu'$ be a
cyclic matching under $P^{\mu,\widetilde{\mu}}$. Then, $\mu'\in S(P^{\mu,\widetilde{\mu}})$.
\end{proposition}

The following proposition says that, given two  Blair-comparable stable matchings, there is a cyclic matching under the reduced preference profile with respect to the Blair-preferred one that is, either the least preferred of the two given stable matchings, or a matching in between the two (with respect to the unanimous Blair order).
\begin{proposition}\label{matching ciclicio arriba del peor}
Let $\mu, \mu' \in S(P)$ with   $\mu \succ_F \mu'$. Then, there is a cyclic matching $\mu_{\sigma}$ under $P^{\mu}$  such
that $\mu \succeq_{F} \mu_{\sigma}\succeq_{F}\mu'$.
\end{proposition}

Finally, we state the last proposition before presenting the algorithm. It says that each stable matching for the original preference profile, different from the firm-optimal stable matching, is always a cyclic matching under a reduced preference profile with respect to  some other stable matching. 

\begin{proposition}\label{todo estable es ciclico}
Let $\mu'\in S(P)\setminus \{\mu_F\}$. Then, there is $\mu\in S(P)$ such that $\mu'$ is a cyclic matching under $P^{\mu}$.
\end{proposition}

\subsection{The Algorithm}\label{sebseccion del algoritmo}
We are now in a position to present our algorithm. Before that, we briefly explain it.
Given a matching market $(F,W,P)$, by the \emph{DA} algorithm we compute the two optimal stable matchings, $\mu_F$ and $\mu_W$. If the two optimal stable matchings are equal, the algorithm stops and the market has only this  stable matching. If they are different,  for the firms' side, we obtain the reduced preference profile with respect to $\mu_F$, $P^{\mu_F}$. In each of the following steps, proceed as follows. For each  reduced preference profile obtained in the previous step, we  compute the following things: (i) each cycle for this profile; (ii) for each cycle, its corresponding cyclic matching; and (iii) for each cyclic matching,  the reduced preference profile with respect to this cyclic matching. The algorithm stops at the step in which all the cyclic matchings computed are equal to the worker optimal stable matching. Formally,

\begin{center}

\begin{tabular}{l l}
\hline \hline
\multicolumn{2}{l}{\textbf{Algorithm:}}\vspace*{10 pt}\\

\textbf{Input} & A many-to-many matching market $(F,W,P)$\\

\textbf{Output} & The set of stable matchings $S(P)$\vspace*{10 pt}\\


\textbf{Step 1} & Find $\mu_F$ and $\mu_W$ (by the \emph{DA}
algorithm) \\ 
& and set $S(P):=\{\mu_F, \mu_W\}$ \\
& \hspace{20 pt}\texttt{IF} $\mu_F=\mu_W$, \\
& \hspace{40 pt}\texttt{THEN} \texttt{STOP}. \\
& \hspace{20 pt}\texttt{ELSE} obtain $P^{\mu_F}$ and continue to Step $2.$  \\
\textbf{Step $\boldsymbol{t}$} &  For each reduced preference profile $P^{\mu}$ obtained
in Step $t-1$, \\
& find all cycles for $P^{\mu}$  and for each cycle obtain its cyclic matching \\
&  under $P^{\mu}$ and include it in $S(P)$.  \\
& \hspace{20 pt}\texttt{IF} each cyclic matching obtained in this step is equal to $\mu_W$, \\
& \hspace{40 pt}\texttt{THEN} \texttt{STOP}. \\
& \hspace{20 pt}\texttt{ELSE} for each cyclic matching $\mu'\neq \mu_W$, obtain the reduced \\
&\hspace{45 pt} preference profile $P^{\mu'}$ and continue to Step $t+1.$ \vspace*{5 pt} \\

\hline \hline
\end{tabular}
\end{center}
\bigskip

Notice that this algorithm stops in a finite number of steps by the finiteness of the market. 
Now, we present the main result of the paper. It states that the firms' optimal stable matching together with all the cyclic matchings obtained by the algorithm encompass the full set of stable matchings.

\begin{theorem}\label{Teorema final}
For a market $(F,W,P)$, the algorithm computes the full set of stable matchings $S(P)$.
\end{theorem}

The following example illustrates the algorithm.

\bigskip

\noindent \textbf{Example 1 (Continued)} \textit{We apply the algorithm to the market of Example \ref{Ejemplo 1}. In what follows, we detail each of its steps:}

 \noindent \textit{\textbf{Step 1}  By the \emph{DA} algorithm, we compute the two optimal stable matchings:} 

\medskip
$
\mu_F=\begin{pmatrix}

f_1 & f_2 & f_3 & \emptyset  \\
 w_1w_2& w_3w_5 &  w_2w_4& w_6 \\
 
\end{pmatrix}
$,
\textit{ and }$
\mu_W=\begin{pmatrix}

f_1 & f_2 & f_3 & \emptyset \\
 w_3w_4& w_2w_5 &  w_1w_2& w_6 \\
 
\end{pmatrix}.$

 \medskip
 
\noindent \textit{Since $ \mu_{F}\neq \mu_W$, we apply the reduction procedure to $ P $  to obtain $ P^{\mu_F} $ which we already computed in Example \ref{Ejemplo 1}}.

\noindent  \textit{\textbf{Step 2} We find all cycles for $ P^{\mu_F} $. There are only two cycles: $ \sigma_{1}=\{( w_1,f_1),(w_4,f_3) \} $ and $ \sigma_{2}=\{ (w_2,f_1),(w_3,f_2) \} $. Their corresponding cyclic matchings are:}\medskip

$
\mu_{\sigma_1}= \begin{pmatrix}

f_1 & f_2 & f_3 & \emptyset  \\
 w_2w_4& w_3w_5 &  w_1w_2& w_6 \\
 
\end{pmatrix}
$,
\textit{and}
 $
\mu_{\sigma_2}= \begin{pmatrix}

f_1 & f_2 & f_3 & \emptyset \\
 w_1w_3& w_2w_5 &  w_2w_4& w_6 \\
 
\end{pmatrix}.$

\medskip
\noindent \textit{Since $ \mu_{\sigma_1}\neq \mu_W$, we apply the reduction procedure to $ P^{\mu_F} $  to obtain the reduced preference profile with respect to $\mu_{\sigma_1}$, $ P^{\sigma_1} $; and since  $ \mu_{\sigma_2}\neq \mu_W$, we apply the reduced preference profile with respect to $\mu_{\sigma_2}$, $ P^{\sigma_2} $. These two profiles are the following:}\medskip

\noindent\begin{tabular}{l}
$P^{\sigma_1}_{f_1}:\boldsymbol{w_2w_4},\boldsymbol{w_3w_4},w_2w_3,w_4,w_3,w_2$\\
$P^{\sigma_1}_{f_2}:\boldsymbol{w_3w_5},\boldsymbol{w_2w_5},w_2w_3,w_2,w_3,w_5$\\
$P^{\sigma_1}_{f_3}:\boldsymbol{w_1w_2},w_1,w_2$\\
$P^{\sigma_1}_{w_1}:\boldsymbol{f_3}$\\
$P^{\sigma_1}_{w_2}:\boldsymbol{f_2f_3},\boldsymbol{f_1f_3},f_1f_2,f_1,f_2,f_3$\\
$ P^{\sigma_1}_{w_3}:\boldsymbol{f_1},\boldsymbol{f_2}$\\
$ P^{\sigma_1}_{w_4}:\boldsymbol{f_1}, f_3$\\
$ P^{\sigma_1}_{w_5}:\boldsymbol{f_2}$\\
$ P^{\sigma_1}_{w_6}:\boldsymbol{\emptyset}$\\
\end{tabular}~~~~~~~~~~~~~~~~~~~
\noindent\begin{tabular}{l}
$P^{\sigma_2}_{f_1}:\boldsymbol{w_1w_3},w_1w_4,\boldsymbol{w_3w_4},w_1,w_4,w_3$\\
$P^{\sigma_2}_{f_2}:\boldsymbol{w_2w_5},w_2,w_5$\\
$P^{\sigma_2}_{f_3}:\boldsymbol{w_2w_4},\boldsymbol{w_1w_2},w_1w_4,w_1,w_2,w_4$\\
$P^{\sigma_2}_{w_1}:\boldsymbol{f_3},\boldsymbol{f_1}$\\
$P^{\sigma_2}_{w_2}:\boldsymbol{f_2f_3},\boldsymbol{f_1f_3},f_1,f_2,f_3$\\
$P^{\sigma_2}_{w_3}:\boldsymbol{f_1}$\\
$P^{\sigma_2}_{w_4}:\boldsymbol{f_1},\boldsymbol{f_3}$\\
$P^{\sigma_2}_{w_5}:\boldsymbol{f_2}$\\
$P^{\sigma_2}_{w_6}:\emptyset$\\
\end{tabular}\medskip

 \noindent \textit{\textbf{Step 3} Lastly, we find all cycles for $ P^{\sigma_1} $ and $ P^{\sigma_2} $. The only cycle for $ P^{\sigma_1} $ is $ \sigma_{2}=\{ (w_2,f_1),(w_3,f_2) \} $. Similarly, the only cycle for $ P^{\sigma_2} $ is  $ \sigma_{1}=\{( w_1,f_1),(w_4,f_3) \} $. Their corresponding cyclic matchings are both equal to $\mu_W$. Then,
 the algorithm stops and $S(P) =\{ \mu_F,\mu_{\sigma_1},\mu_{\sigma_2},\mu_W \}$.\hfill $\Diamond$}

\section{Concluding Remarks}\label{Concludings}

 For a many-to-many matching market in which agents' preferences satisfy substitutability and \emph{LAD}, we presented an algorithm to compute the full set of stable matchings. Our approach extends the notion of cycles and cyclic matchings presented in the classic book of  \cite{roth1992two}.  Given  any stable matching $\mu$, each adjacent stable matching $\mu'$ 
 is obtained as a cyclic matching under the reduced preference profile $P^{\mu}$.\footnote{Stable matchings $\mu$ and $\mu'$ are \emph{adjacent} if $\mu\succ_F\mu'$ and there is no other stable matching $\mu''$ such that $\mu \succ_F \mu'' \succ_F \mu'$.} Even though our results make no use of the lattice structure of the stable set, our algorithm  travels through this lattice from the firm-optimal to the worker-optimal stable matching, finding all stable matchings in between. 

It is known that the complexity of implementation of any algorithm that evaluates a choice function for substitutable preferences is exponential (a choice function requires exponential queries to a substitutable preferences relation). However, when preferences are substitutable, computer scientists usually assume the existence of an artificial \emph{oracle} to the choice function: in every iteration of an algorithm, each agent can query the oracle to determine its favorite subset of opposite sided agents available \citep[see ][for more details]{deng2017complexity}. With the assumption of an oracle, our algorithm can be run in polynomial time. 

A  paper closely related to ours is \cite{martinez2004algorithm}, that claims to compute the full set of many-to-many stable matchings. An important difference between the algorithm of \cite{martinez2004algorithm} and ours is that theirs is based on the one-to-one algorithm presented by \cite{mcvitie1971stable}. The \emph{DA} algorithm  must be applied to a reduced preference profile in \emph{each} step of the algorithm of  \cite{martinez2004algorithm}, while in our algorithm we use the \emph{DA} algorithm only twice (to calculate the firm-optimal  and worker-optimal stable matchings in the first step) and afterward we only seek for cycles in a reduced preference profile and compute their corresponding cyclic matchings. Another difference is  that  \cite{martinez2004algorithm} only assume substitutability on agents' preferences, while we assume in addition \emph{LAD}.

Next, we provide an example that shows that the algorithm in \cite{martinez2004algorithm} has an error (the algorithm does not compute the full set of stable matchings).
Before presenting this example, we roughly explain how their algorithm works. Let $(F,W,P)$ be a matching market. By using the \emph{DA} algorithm, compute $\mu_F$ and $\mu_W$ and set $S^{\star}(P)=\{\mu_F,\mu_W\}$. In Step $1$, for each pair $(f,w)$ such that $w\in \mu_F(f)\setminus \mu_W(f)$, (i) compute the  $w$--truncation of $P_f$ and consider the new preference profile $P^{(f,w)}$ obtained from $P$ by replacing $P_f$ by the  $w$--truncation of $P_f$;\footnote{See Definition \ref{defino trucacion} in the Appendix.}  (ii) compute, by the \emph{DA} algorithm, the firm-optimal stable matching  for the new market $(F,W,P^{(f,w)})$, denoted by  $\mu_F^{(f,w)}$; (iii) if  $C_{w'}(\mu_F(w')\cup \mu_F^{(f,w)}(w'))=\mu_F^{(f,w)}(w')$ for each $w'\in W$, then add $\mu_F^{(f,w)}$ to $S^{\star}(P)$. In Step $t$, for each  matching added  to $S^{\star}(P)$ in Step $t-1$, repeat items (i), (ii),
and (iii) of Step $1$ for each pair $(f,w)$ such that $w$ is matched to $f$ under this new matching but not under the original worker-optimal stable matching.
The algorithm stops in the step in which no matching is added  to $S^{\star}(P)$. \cite{martinez2004algorithm} wrongly state that  $S^{\star}(P)=S(P).$
 
Now, we are in a position to  present the example\footnote{This example was provided to one of the authors of this paper by Xuan Zhang.} showing that: (i) the algorithm of \cite{martinez2004algorithm} stops before computing all stable matchings, and (ii) our algorithm computes the full set of stable matchings.
 
 \begin{example}\label{ejemplo 2}
 Let $(F,W,P)$ be a one-to-one matching market in which $F=\{f_1,f_2,f_3,f_4\}$, $W=\{w_1,w_2,w_3,w_4 \}$, and the preference profile is given by:
 
\begin{center}
\noindent\begin{tabular}{lcl}
$P_{f_1}:w_2,w_1,w_3,w_4$&~~~~~~& $ P_{w_1}:f_2,f_1,f_4,f_3$\\
$P_{f_2}:w_4,w_2,w_3,w_1$&~~~~~~& $ P_{w_2}:f_4,f_3,f_2,f_1$\\
$P_{f_3}:w_4,w_2,w_3,w_1$&~~~~~~& $ P_{w_3}:f_3,f_1,f_4,f_2$\\
$P_{f_4}:w_3,w_1,w_4,w_2$&~~~~~~& $ P_{w_4}:f_1,f_3,f_4,f_2$\\
\end{tabular}
\end{center}

\noindent Agents' preferences in a one-to-one matching market satisfy substitutability and LAD because they are linear orderings among single agents. For this market, there are three stable matchings: 

\begin{center}
 $
\mu_F=\begin{pmatrix}

f_1 & f_2 & f_3 & f_4  \\
 w_1& w_2 &  w_4& w_3 \\
 
\end{pmatrix}
$, $
\mu=\begin{pmatrix}

f_1 & f_2 & f_3 & f_4  \\
 w_3& w_1 &  w_4& w_2 \\
 
\end{pmatrix}
$,
 and $
\mu_W=\begin{pmatrix}

f_1 & f_2 & f_3 & f_4\\
 w_4& w_1&  w_3& w_2 \\
 
\end{pmatrix}.$
\end{center}

\noindent Following the algorithm of \cite{martinez2004algorithm}, the pairs $(f,w)$ such that $w\in \mu_F(f)\setminus \mu_W(f)$ are: $(f_1,w_1),(f_2,w_2),(f_3,w_4)$ and $(f_4,w_3)$. For each of these pairs $(f,w)$, the firm-optimal stable matching for the $w$--truncation of $P_f$ are:
\begin{center}
$
\mu^{(f_1,w_1)}_F=\begin{pmatrix}

f_1 & f_2 & f_3 & f_4  \\
 w_3& w_2 &  w_4& w_1 \\
 
\end{pmatrix}
$,
$
\mu^{(f_2,w_2)}_F=\begin{pmatrix}

f_1 & f_2 & f_3 & f_4  \\
 w_2& w_1 &  w_4& w_3 \\
 
\end{pmatrix}
$,
\end{center}

\begin{center}
$
\mu^{(f_3,w_4)}=\begin{pmatrix}

f_1 & f_2 & f_3 & f_4  \\
 w_1& w_4 &  w_2& w_3 \\
 
\end{pmatrix}
$, and
$
\mu^{(f_4,w_3)}=\begin{pmatrix}

f_1 & f_2 & f_3 & f_4  \\
 w_1& w_3 &  w_4& w_2 \\
 
\end{pmatrix}.
$
\end{center}

\noindent Notice that 

$C_{w_1}(\mu_F(w_1)\cup \mu_F^{(f_1,w_1)}(w_1))=C_{w_1}(\{f_1,f_4\})=\{f_1\}\neq \mu_F^{(f_1,w_1)}(w_1),$

$C_{w_2}(\mu_F(w_2)\cup \mu_F^{(f_2,w_2)}(w_2))=C_{w_2}(\{f_2,f_1\})=\{f_2\}\neq \mu_F^{(f_2,w_2)}(w_2),$

$C_{w_4}(\mu_F(w_4)\cup \mu_F^{(f_3,w_4)}(w_4))=C_{w_4}(\{f_3,f_2\})=\{f_3\}\neq \mu_F^{(f_3,w_4)}(w_4),$ and

$C_{w_3}(\mu_F(w_3)\cup \mu_F^{(f_4,w_3)}(w_3))=C_{w_3}(\{f_4,f_2\})=\{f_4\}\neq \mu_F^{(f_4,w_3)}(w_3).$
\medskip

\noindent Thus, the algorithm does not incorporate any matching to $S^{\star}(P)$ and, therefore, stops without computing  stable matching $\mu$.

\noindent Now we show how our algorithm computes all of these three stable matchings. Once we compute $\mu_F$ and $\mu_W$ by the DA algorithm, the reduced preference profile $P^{\mu_F}$ is given by:

 \begin{center}
 \noindent\begin{tabular}{lcl}
$P^{\mu_F}_{f_1}:w_1,w_3,w_4$&~~~~~~& $ P^{\mu_F}_{w_1}:f_2,f_1$\\
$P^{\mu_F}_{f_2}:w_2,w_1$&~~~~~~& $ P^{\mu_F}_{w_2}:f_4,f_3,f_2$\\
$P^{\mu_F}_{f_3}:w_4,w_2,w_3$&~~~~~~& $ P^{\mu_F}_{w_3}:f_3,f_1,f_4$\\
$P^{\mu_F}_{f_4}:w_3,w_2$&~~~~~~& $ P^{\mu_F}_{w_4}:f_1,f_3$\\
\end{tabular}

\end{center}  

\noindent It is easy to check that there is only one cycle for $P^{\mu_F}$, $\sigma_{1}=\{ (w_1,f_1),(w_3,f_4),(w_2,f_2) \} $. Its corresponding cyclic matching is $\mu_{\sigma_1}=\mu$. Now, the reduced preference profile $P^{\mu_{\sigma_1}}$ is given by:
\begin{center}
 \noindent\begin{tabular}{lcl}
$P^{\mu_F}_{f_1}:w_3,w_4$&~~~~~~& $ P^{\mu_F}_{w_1}:f_2$\\
$P^{\mu_F}_{f_2}:w_1$&~~~~~~& $ P^{\mu_F}_{w_2}:f_4$\\
$P^{\mu_F}_{f_3}:w_4,w_3$&~~~~~~& $ P^{\mu_F}_{w_3}:f_3,f_1$\\
$P^{\mu_F}_{f_4}:w_2$&~~~~~~& $ P^{\mu_F}_{w_4}:f_1,f_3$\\
\end{tabular}
\end{center}

\noindent Finally, it is easy to check that there is only one cycle for $P^{\mu_{\sigma_1}}$, $\sigma_{2}=\{ (w_3,f_1),(w_4,f_3)\} $. Its corresponding cyclic matching is $\mu_{\sigma_2}=\mu_W$. In this way, our algorithm computes the full set of stable matchings for the market $(F,W,P).$ \hfill $\Diamond$
 
 \end{example}
It seems that the problem with the algorithm in \cite{martinez2004algorithm} is that it is ill-posed in the following way. There are matchings that the algorithm computes that are not stable in the original preferences.  Because of this, the algorithm dismisses them. But those matchings turn out to be crucial to find new stable matchings.
 For instance, in Example \ref{ejemplo 2}, the matching $\mu^{(f_1,w_1)}_F$ is unstable for the original preferences, and the algorithm in \cite{martinez2004algorithm} dismisses it. However if we truncate the preference profile $P^{(f_1,w_1)}$ with the pair $(f_4,w_1)$ and obtain the firm optimal matching for this new truncated profile, we obtain matching $ \mu$ which is stable in the original preference profile and is never computed by the algorithm (as shown in the previous example).

 \appendix
 \section{Appendix}\label{appendix}

In order to prove Lemma \ref{preferenceias resucidas sustituibles}, we first define a $w$--truncation of preference $P_f$ and adapt  a lemma of \cite{martinez2004algorithm} to our setting.
\begin{definition}[\citealp{martinez2004algorithm}]\label{defino trucacion}
We say that  the preference $P_f^{w}$ is the $w$--truncation of $P_f$ if:
\begin{enumerate}[(i)]
\item All sets containing $w$ are unacceptable to $f$ according to $P_f^w$. That is, if $w\in S$ then $\emptyset P_f^w S.$
\item The preferences $P_f$ and $P_f^w$ coincide on all sets that do not contain $w$. That is, if $w\notin S_1 \cup S_2$ then $S_1 P_f S_2$ if and only if $S_1P_f^w S_2$.
\end{enumerate}Similarly, we define $P^f_w$ as  an $f$--truncation of $P_w$.
\end{definition}
\begin{remark}\label{remark de lad en truncacion}
Given a $w$--truncation of $P_f$ and any subset of workers $S$,  $C_f(S\setminus \{w\})=C_f^w(S).$ Similarly, given a $f$--truncation of $P_w$ and any subset of firms $S$,  $C_w(S\setminus \{f\})=C_w^f(S).$
\end{remark}

\begin{lemma}\label{truncacion es sustituible}
Let $f\in F$ and $w\in W$ with their respective preference relations $P_f$ and $P_w$. If $P_f$ is substitutable  and satisfies LAD, then $P^w_{f}$ is substitutable and satisfies LAD. Similarly, if $P_w$ is substitutable and satisfies LAD, then $P^f_{w}$ is substitutable and satisfies LAD.
\end{lemma}
\begin{proof} Let $f\in F$, $w\in W$, and $P_f$ be substitutable and satisfies \emph{LAD}. Let $P^w_{f}$ be the $w$--truncation of $P_f$. We only prove that if $P_f$ is substitutable and satisfies \emph{LAD}, then $P^w_{f}$ is substitutable  and satisfies \emph{LAD}. The other implication is analogous. 
To see that $P^w_{f}$ is substitutable, let $ \widetilde{w}, w'\in S$ be arbitrary and assume that $\widetilde{w}\in C_f^w(S).$\footnote{Denote by  $C_f^w(S)$ to $f$'s most preferred subset of $S$ according to the $w$--truncation of $P_f$.}  If $w\notin S$, then $\widetilde{w}\in C^w_f(S\setminus \{w'\})$ because $C^w_f(S)=C_f(S)$, $C^w_f(S\setminus \{w'\})=C_f(S\setminus \{w'\})$, and because of the substitutability of $P_f$. If $w\in S$, then we have that $C^w_f(S)=C_f(S\setminus \{w\})$; therefore, by assumption $\widetilde{w}\in C_f(S\setminus \{w'\}).$ By the substitutability of $P_f$, we have that $\widetilde{w}\in C_f([S\setminus \{w\}]\setminus \{w'\}).$ But, the equality $C_f([S\setminus \{w\}]\setminus \{w'\})=C^w_f(S\setminus \{w'\})$ implies that $\widetilde{w} \in C^w_f(S\setminus \{w'\})$. Therefore, $P^w_{f}$ is substitutable.

To see that $P^w_{f}$ satisfies \emph{LAD}, let $S'$ and $S$ be two subsets of workers such that $S'\subset S$. Note that, $S'\setminus\{w\} \subset S\setminus\{w\}.$ Then, by Remark \ref{remark de lad en truncacion} and the fact that $P_{f}$ satisfies \emph{LAD} we have,
$$|C^w_f(S')|= |C_f(S'\setminus \{w\})|\leq |C_f(S\setminus \{w\})|=|C^w_f(S)|.$$
Therefore, $P^w_{f}$ satisfies \emph{LAD}.
\end{proof}

\bigskip

\noindent \begin{proof}[Proof of Lemma \ref{preferenceias resucidas sustituibles}]W.l.o.g. assume that agent $a$ is a firm, say $f\in F$.
Let $P_f$ be a substitutable preference that satisfies \emph{LAD}. Let $\widetilde{W}_f$ be the set of  workers selected in  Step 1 (a), Step 2 (a) or Step 3 of the reduction procedure for firm $f$. Take any $w \in \widetilde{W}_f$ and consider the $w$--truncation $P_f^w$. By Lemma \ref{truncacion es sustituible}, preference $P_f^w$ is substitutable  and satisfies \emph{LAD}. Now, any take $w'\in \widetilde{W}_f\setminus \{w\}$ and consider the $w'$--truncation of $P_f^w$. Again by Lemma \ref{truncacion es sustituible}, the $w'$--truncation of $P_f^w$ is substitutable  and satisfies \emph{LAD}. Continuing in the same way for each worker of $\widetilde{W}_f$ not yet considered, we construct the corresponding truncation of the previously obtained truncated preference. By Lemma \ref{truncacion es sustituible}, each one of these truncated preferences is substitutable  and satisfies \emph{LAD}.  By the finiteness of the set $\widetilde{W}_f$, this process will end. Moreover, by definition of  $\widetilde{W}_f$,  the last truncated preference obtained in this process is $P_f^{\mu,\widetilde{\mu}}$. Therefore, preference $P_f^{\mu,\widetilde{\mu}}$ is substitutable  and satisfies \emph{LAD}. 
\end{proof}

In order to prove Theorem \ref{estable original sii estable en el reducido M-to-M}, we first show in Lemma \ref{individual dentro de mu y mu tilde} that,  under certain conditions, individual rationality of a matching under the original preference profile is equivalent to individual rationality  under a reduced preference profile. As we said before the statement of Theorem \ref{estable original sii estable en el reducido M-to-M}, \emph{LAD} is no required for this result.

\begin{lemma}\label{individual dentro de mu y mu tilde}
Let $\mu, \widetilde{\mu} \in S(P)$ with $\mu \succ_F \widetilde{\mu}$ and let $\mu'$ be a matching. The matching  $\mu'$ is individually rational under $P$ with $\mu \succeq_{F}\mu'\succeq_{F} \widetilde{\mu}$ and  $\widetilde{\mu} \succeq_{W}\mu'\succeq_{W} \mu$  if and only if $\mu'$ is individually rational under $P^{\mu,\widetilde{\mu}}.$\footnote{Recall that $\succeq_F$ and $\succeq_W$ are dual orders only in the set of stable matchings, so both $\mu \succeq_{F}\mu'\succeq_{F} \widetilde{\mu}$ and  $\widetilde{\mu} \succeq_{W}\mu'\succeq_{W} \mu$ need to be required.}
\end{lemma}
\begin{proof}
Let $ \mu, \widetilde{\mu}\in S(P)$ with $\mu \succ_F \widetilde{\mu}$ and let $\mu'$ be a matching.  

\noindent $(\Longrightarrow)$ Assume that the matching $ \mu'  $ is individually rational under $P$ with  $\mu \succeq_{F}\mu'\succeq_{F} \widetilde{\mu}$ and  $\widetilde{\mu} \succeq_{W}\mu'\succeq_{W} \mu$. We claim that $ \mu'(f) $ and $ \mu'(w) $ are not eliminated in the reduction to obtain $P^{\mu,\widetilde{\mu}}$ for each $ f\in F $ and $ w\in W $. Since $\mu \succeq_{F}\mu'\succeq_{F} \widetilde{\mu}$, we have $ \mu(f)= C_{f}(\mu(f)\cup
\mu'(f))$ and $ \mu'(f)= C_{f}(\widetilde{\mu}(f)\cup
\mu'(f))$ for each $ f\in F $ . Moreover, since  $\widetilde{\mu} \succeq_{W}\mu'\succeq_{W} \mu$, we have
$\widetilde{\mu}(w)= C_{w}(\mu'(w)\cup
\widetilde{\mu}(w))$ and $ \mu'(w)= C_{w}(\mu'(w)\cup
\mu(w))$
for each $ w\in W $. Therefore, $ \mu'(f) $ and $ \mu'(w) $ are not eliminated at Step 1, and Step 2 
of the reduction procedure. 
Let $ (f,w) $ be a pair  assigned in $\mu'  $. Since $\mu'$ is individually rational, the pair $(f,w)$ is mutually acceptable under $P$. Moreover, since $ \mu'(f) $ and $ \mu'(w) $ were not eliminated at Step 1, or Step 2, then $(f,w)$ is mutually acceptable under $P^{\mu,\widetilde{\mu}}$.
Thus, no pair of agents assigned in $ \mu' $ is eliminated in Step 3 of the reduction procedure. Then,
\begin{equation*}
C^{\mu,\widetilde{\mu}}_{f}(\mu'(f))=C_{f}(\mu'(f))=\mu'(f)
\end{equation*}
and 
\begin{equation*}
C^{\mu,\widetilde{\mu}}_{w}(\mu'(w))=C_{w}(\mu'(w))=\mu'(w).
\end{equation*}
Therefore, $ \mu' $ is an individuality rational matching under $ P^{\mu,\widetilde{\mu} } $.

\noindent $(\Longleftarrow )$  Assume that the matching $ \mu'$ is individually rational under $P^{\mu,\widetilde {\mu}}.$ The definition of reduced preference $P^{\mu,\widetilde {\mu}}$ implies that $\mu \succeq_{F}\mu'\succeq_{F} \widetilde{\mu}$ and  $\widetilde{\mu} \succeq_{W}\mu'\succeq_{W} \mu.$ Let $f\in F.$ Notice that, by the reduction procedure, if $ w\in  C^{\mu,\widetilde{\mu}}_{f}(\mu'(f))$ then, $ w\in  C_{f}(\mu'(f))$. Since  $\mu'(f)= C^{\mu,\widetilde{\mu}}_{f}(\mu'(f)) \subseteq C_{f}(\mu'(f))\subseteq \mu'(f)$, we have that $ C_{f}(\mu'(f))=\mu'(f)$. Similarly,  $C_{w}(\mu'(w))=\mu'(w)$ for each $w \in W$. Therefore,  $\mu' $ is an individually rational matching under $P$.
\end{proof}

\bigskip

\noindent \begin{proof}[Proof of Theorem \ref{estable original sii estable en el reducido M-to-M}]Let $\mu,\widetilde{\mu}\in S(P)$ with $\mu\succ_{F}\widetilde{\mu}$.

 \noindent $(\Longleftarrow)$
Let $\mu'\in S(P^{\mu, \widetilde{\mu}}).$
By Lemma \ref{individual dentro de mu y mu tilde}, we have that $\mu'$ is individually rational under $P$. Assume that $\mu'\notin S(P)$. Thus, there is a blocking pair of $\mu'$ under $P,$ i.e.,  there is  $(f,w) \in F \times W$ such that $w\notin \mu'(f),~ w\in C_f(\mu'(f)\cup \{w\})$ and $f\in C_w(\mu'(w)\cup \{f\}).$

\noindent \textbf{Claim: neither $\boldsymbol{C_f(\mu'(f)\cup \{w\})$ nor $C_w(\mu'(w)\cup \{f\})}$ is eliminated in the reduction procedure}. First, assume w.l.o.g. that $C_f(\mu'(f)\cup \{w\})$ is eliminated in Step 1 the reduction procedure. There are two cases to consider:

\noindent \textbf{Case 1: $\boldsymbol{ C_f(\mu'(f)\cup \{w\})\nsucceq_f \mu(f)}$}. Thus, there are $W'$ and  $\widetilde{w}$ such that $\widetilde{w}\in W'\setminus \mu(f)$, $\widetilde{w}\in C_f(\mu'(f)\cup \{w\})$ and $W'=C_f(W' \cup \mu(f)).$ Then, if $\widetilde{w}\in \mu'(f)$, $\mu'(f)$ is eliminated in Step 1 of the reduction procedure. Therefore, $\mu'(f)\neq C_f^{\mu, \widetilde{\mu}}(\mu'(f)),$ and $\mu'$ is not individually rational under  $P^{\mu, \widetilde{\mu}}$, contradicting Lemma \ref{individual dentro de mu y mu tilde}. If $\widetilde{w}=w$, $w\in C_f(W' \cup \mu(f))$ and, by substitutability, 
\begin{equation}\label{queseyo0}
w\in C_f(\mu(f)\cup \{w\}).
\end{equation}
Moreover, by definition of Blair's partial order and \eqref{propiedad 1} 
\begin{equation}\label{queseyo1}
C_w(\mu'(w)\cup \{f\}) \succeq_w C_w(\mu'(w)).
\end{equation}
 Since $\mu'$ is individually rational under $P^{\mu, \mu'},$ $\mu'$ is individually rational under $P$ and $\widetilde{\mu} \succeq_{W}\mu'\succeq_{W} \mu$ by Lemma \ref{individual dentro de mu y mu tilde}. Then, $\mu'(w)=C_w(\mu'(w)) \succeq_w \mu(w)$ and, by \eqref{queseyo1} and  transitivity of $\succeq_w$, we have $
C_w(\mu'(w)\cup \{f\}) \succeq_w\mu(w).$  Thus, $C_w(\mu'(w)\cup \{f\})=C_w(\mu(w) \cup C_w(\mu'(w)\cup \{f\})).$ Applying \eqref{propiedad 1}, we have $C_w(\mu(w) \cup C_w(\mu'(w)\cup \{f\}))=C_w(\mu(w) \cup \mu'(w)\cup \{f\})$.
Recall that $(f,w)$ is a blocking pair for $\mu'$ under $P$. Hence, $f\in C_w(\mu'(w)\cup \{f\})=C_w(\mu(w) \cup \mu'(w)\cup \{f\}).$ Since $w=\widetilde{w}\in W' \setminus \mu(f)$, $f\notin \mu(w)$. Thus, by substitutability,  $f\in C_w(\mu(w) \cup \mu'(w)\cup \{f\})$ implies that 
\begin{equation}\label{queseyo3}
f \in C_w(\mu(w)\cup \{f\}).
\end{equation}
Furthermore, since $w=\widetilde{w}$, $w \notin \mu(f)$. This, together  with  \eqref{queseyo0} and \eqref{queseyo3} imply that $(f,w)$ is a blocking pair for $\mu$ under $P$. This is a contradiction.

\noindent \textbf{Case 2: $\boldsymbol{ C_f(\mu'(f)\cup \{w\}) \succeq_f \mu(f)}$}. Since we assume that $(f,w)$ is a blocking pair for $\mu'$ under $P$,  $w\in C_f(\mu'(f)\cup \{w\}).$ By this case's hypothesis, $w\in  C_f(\mu(f) \cup C_f(\mu'(f)\cup \{w\}))$.  By \eqref{propiedad 1},  $w\in C_f(\mu(f) \cup \mu'(f) \cup \{w\})=C_f(C_f(\mu(f) \cup \mu'(f)) \cup \{w\})$. Since $\mu \succeq_F \mu'$, 
\begin{equation}\label{queseyo4}
 w\in C_f(\mu(f) \cup \{w\}).
\end{equation}

Now, we claim  that $w\notin \mu(f)$. First, note that $f\in C_w(\mu'(w)\cup \{f\})$ and $f\notin \mu'(w)$ implies that $C_w(\mu'(w)\cup \{f\})\succ_w \mu'(w).$ Second,  $\mu'\in S(P^{\mu,\widetilde{\mu}})$ implies, by Lemma \ref{individual dentro de mu y mu tilde}, $\mu'\succeq_W \mu$. Thus,  $ \mu'(w)=C_w(\mu'(w)\cup \mu(w))$ for each $w\in W.$  Lastly, since $f\notin \mu'(w)$ and  assuming  that $f\in \mu(w)$, we conclude that $$ \mu'(w)=C_w(\mu'(w)\cup \mu(w))\succeq_w C_w(\mu'(w)\cup \{f\})\succ_w \mu'(w),$$  and this is a contradiction. Then, 
\begin{equation}\label{queseyo5}
w\notin \mu(f).
\end{equation} 
Moreover, by the same argument used to obtain \eqref{queseyo3},    $f\in C_w(\mu'(w)\cup \{f\})$ implies that 
\begin{equation}\label{queseyo6}
f\in  C_w(\mu(w)\cup \{f\}).
\end{equation} 
Hence, by \eqref{queseyo4}, \eqref{queseyo5}, and \eqref{queseyo6}, $(f,w)$ is a blocking pair for $\mu$ under $P$, and this is a contradiction. 
Therefore, by Case 1 and Case 2, $C_f(\mu'(f)\cup \{w\})$ is not eliminated in Step 1. 

Second, assume w.l.o.g. that $C_f(\mu'(f)\cup \{w\})$ is eliminated in Step 2 the reduction procedure.  Note that this cannot happen, because $C_f(\mu'(f)\cup \{w\})\succeq_f \mu'(f) \succeq_f\widetilde{\mu}(f)$ for each $f\in F.$
By a symmetrical argument, we have that $C_w(\mu'(w)\cup \{f\})$ can not be eliminated in Step 1 or Step 2 either. 

Now, we show that  neither $C_f(\mu'(f)\cup \{w\})$ nor $C_w(\mu'(w)\cup \{f\})$ is eliminated in Step 3.
Assume w.l.o.g.  that $C_f(\mu'(f)\cup \{w\})$ is eliminated in Step 3. Thus,  there is $\widetilde{w}\in C_f(\mu'(f)\cup \{w\})$ such that $ \widetilde{w} $ is not acceptable for $ f $ after Steps 1 and 2 are performed. Note that this implies that $ C_f^{\mu,\mu'}( \{\widetilde{w}\})\neq\{\widetilde{w}\}$. By definition of $ C_f $ we have $ C_f(\mu'(f)\cup \{w\})\subseteq  \mu'(f)\cup \{w\} $. Thus, $ \widetilde{w}\in \mu'(f) $ or $ \widetilde{w}=w $. If $ \widetilde{w}\in \mu'(f) $, since $ \mu' $ is individually rational under $ P^{\mu,\widetilde{\mu}} $, $\widetilde{w}\in\mu'(f)=C_f^{\mu,\widetilde{\mu}}( \mu'(f))$ and, by substitutability,  $\widetilde{w}\in C_f^{\mu,\widetilde{\mu}}( \{\widetilde{w}\})\neq \{\widetilde{w}\},$ which is absurd. Therefore,  $ \widetilde{w}=w $. Since $ (f,w) $ is a blocking pair of $ \mu' $ under $ P $, $w\in C_f(\mu'(f)\cup \{w\}) $.  Since $ w $ is not acceptable for $ f $ after Step 1 and Step 2, this implies that any set that contains agent $ w $ is removed from $f$'s preference list at Step 1 or Step 2. Thus,  $ C_f(\mu'(f)\cup \{w\}) $ is removed from $f$'s preference list in Step 1 or Step 2, and this is a contradiction. Therefore, $C_f(\mu'(f)\cup \{w\})$ is not eliminated in Step 3. A similar argument proves that $C_w(\mu'(w)\cup \{f\})$ is not eliminated either in Step 3.  This completes the proof of the Claim.

In order to finish the proof, since by the Claim neither $C_f(\mu'(f)\cup \{w\})$ nor $C_f(\mu'(w)$ $\cup \{f\})$ is eliminated by the reduction procedure, we have that $C_f(\mu'(f)\cup \{w\})=C^{\mu,\widetilde{\mu}}_f(\mu'(f)\cup \{w\})$ and  $C_f(\mu'(w)\cup \{f\})=C^{\mu,\widetilde{\mu}}_f(\mu'(w)\cup \{f\})$. Then, $(f,w)$ is a blocking pair for $\mu'$ under  $P^{\mu,\widetilde{\mu}}$, and this is a contradiction.  Therefore, $\mu'\in S(P).$
\medskip

\noindent $(\Longrightarrow)$
Let $\mu'\in S(P)$ with $\mu\succeq_{F}\mu'\succeq_{F}\widetilde{\mu}$. This implies that $\widetilde{\mu}\succeq_{W}\mu'\succeq_{W}\mu$. By Lemma \ref{individual dentro de mu y mu tilde}, we have that $\mu'$ is individually rational in $P^{\mu,\widetilde{\mu}}$. Assume that $\mu'\notin S(P^{\mu,\widetilde{\mu}})$. Thus, there is a pair $(f,w) \in F \times W$ such that $w\notin \mu'(f),~ w\in C^{\mu,\widetilde{\mu}}_f(\mu'(f)\cup \{w\})$ and $f\in C^{\mu,\widetilde{\mu}}_f(\mu'(w)\cup \{f\}).$ By the reduction procedure   $~ w\in C_f(\mu'(f)\cup \{w\})$ and $f\in C_f(\mu'(w)\cup \{f\}). $ Therefore the pair $(f,w)$ blocks $ \mu' $ under $P$, and this is a  contradiction of $\mu'\in S(P)$. Thus, $\mu'\in S(P^{\mu,\widetilde{\mu}}).$
\end{proof}

\bigskip

\noindent \begin{proof}[Proof of Proposition \ref{existencia de ciclo}]$(\Longrightarrow)$ This implication is straightforward from Definition \ref{defino cyclo}, since there is a cycle only if there is a firm $f$ such that 
$\mu(f)\neq \widetilde{\mu}(f)$ under $P^{\mu,\widetilde{\mu}}$.

\noindent $(\Longleftarrow )$ Assume that $\mu \neq \widetilde{\mu}$. We construct a bipartite oriented the digraph $D^{\mu, \widetilde{\mu}}$ with sets of nodes 
\[
V_{1} =\left\{ \left( w,f\right) \in W\times F:w\in \mu(f) \setminus \widetilde{\mu}(f) \right\} 
\] and 
\[
V_{2} =(F \times W) \setminus \{(f,w) : (w,f) \in V_1\}. 
\]
Since $\mu \neq \widetilde{\mu},$ both $V_1$ and $V_2$ are non-empty. The oriented arcs are defined as follows. There is and arc from $(w,f) \in V_{1}$ to $(f',w') \in V_{2}$ if $$f=f' \text{  and  }  C^{\mu, \widetilde{\mu}}_{f}( W\setminus \{w\})= \left(\mu(f) \setminus \{w\}\right)\cup \{w'\}.$$
There is an arc from $(f',w') \in V_2$ to $(w,f) \in V_1$ if $$w'=w \text{  and  } C^{\mu, \widetilde{\mu}}_{w}\left(\mu(w)\cup \{f'\}\right)=\left(\mu(w)\setminus \{f\}\right)\cup \{f'\}.$$
It is easy to see that there is an oriented cycle in the digraph $D^{\mu, \widetilde{\mu}}$ if and only if there is  a cycle for preference $P^{\mu, \widetilde{\mu}}.$ In fact, if $\{(w_1, f_1), (f_1, w_2), (w_2, f_2), (f_2, w_3), \ldots, (w_r, f_r), (f_r, w_1)\}$ is a cycle for $D^{\mu, \widetilde{\mu}},$ then $\{w_1, f_1, w_2, f_2, \ldots, w_r, f_r\}$ is a cycle for $P^{\mu, \widetilde{\mu}}.$ Assume that there is no cycle for $P^{\mu, \widetilde{\mu}}.$ Then, there is no cycle in digraph $D^{\mu, \widetilde{\mu}}.$ Let $p$  be a maximal path in $D^{\mu, \widetilde{\mu}}.$ There are two cases to consider: 

\noindent \textbf{Case 1: the terminal node  $\boldsymbol{(w,f)$ of  $p$ belongs to $V_1}$}. Then $w \in \mu(f)\setminus \widetilde{\mu}(f)$ and there is no $w' \in W$ such that $w' \notin \mu(f) \setminus \widetilde{\mu}(f)$ and  $w' \in C_{P_{f}^{\mu,\widetilde{\mu}}}( W\setminus
\{w\}).$ Therefore, $C_{f}^{\mu,\widetilde{\mu}}( W\setminus
\{w\}) \subsetneq C_{f}^{\mu,\widetilde{\mu}}( W) =\mu(f).$ By \emph{LAD}, 
\begin{equation}\label{pruebadigrafo1}
|C_{f}^{\mu,\widetilde{\mu}}( W\setminus
\{w\})| < |\mu(f)|.
\end{equation}
Moreover, since $w \in \mu(f)\setminus \widetilde{\mu}(f),$ we have $\widetilde{\mu}(f) \subseteq W \setminus \{w\}.$ Thus, $\widetilde{\mu}(f)=C_{f}^{\mu,\widetilde{\mu}}(\widetilde{\mu}(f)) \subseteq C_{P_{f}^{\mu,\widetilde{\mu}}}( W\setminus \{w\})$ and, by \emph{LAD},  
\begin{equation}\label{pruebadigrafo2}
|\widetilde{\mu}(f)| \leq | C_{f}^{\mu,\widetilde{\mu}}( W\setminus \{w\})|.
\end{equation}
By the Rural Hospitals Theorem,\footnote{The \textit{Rural Hospitals Theorem} states that,  under substitutability and \emph{LAD}, each agent is matched with the same number of partners in every stable matching. That is, $|\mu(a)|=|\mu'(a)|$ for each $\mu,\mu'\in S(P)$ and for each $a\in F\cup W$ \citep[see][for more details]{alkan2002class}.} $|\mu(f)|=|\widetilde{\mu}(f)|.$ This, together with \eqref{pruebadigrafo1} and \eqref{pruebadigrafo2} implies that $|\widetilde{\mu}(f)| \leq | C_{f}^{\mu,\widetilde{\mu}}( W\setminus \{w\})|<|\mu(f)|=|\widetilde{\mu}(f)|,$ which is absurd.

\noindent \textbf{Case 2: the terminal node  $\boldsymbol{(f',w)$ of  $p$ belongs to $V_2}$}. Then, $f' \notin \mu(w)\setminus \widetilde{\mu}(w).$ First, we claim  that $| C_{w}^{\mu,\widetilde{\mu}}( \mu(w) \cup \{f'\}) | = |\mu(w)|.$  Since  $C_{w}^{\mu,\widetilde{\mu}}(F)=\widetilde{\mu}(w)$ by Remark \ref{remark de reduccion M-M} (i) and $\mu(w) \cup \{f'\} \subseteq F,$ by \emph{LAD} it follows that
\begin{equation}\label{ecu2 lema 5}
|\widetilde{\mu}(w)| \geq | C_{w}^{\mu,\widetilde{\mu}}( \mu(w) \cup \{f'\}) |.
\end{equation}
Furthermore, by \emph{LAD} and individual rationality of $\mu,$
\begin{equation}\label{ecu1 lema 5}
| C_{w}^{\mu,\widetilde{\mu}}( \mu(w) \cup \{f'\}) | \geq | C_{w}^{\mu,\widetilde{\mu}}( \mu(w)) |= |\mu(w)|.
\end{equation}
Assume $| C_{w}^{\mu,\widetilde{\mu}}( \mu(w) \cup \{f'\}) | > |\mu(w)|.$ By \eqref{ecu1 lema 5} and \eqref{ecu2 lema 5}, it follows that $|\widetilde{\mu}(w)| > |\mu(w)|.$  This contradicts the Rural Hospitals Theorem. Therefore, $| C_{w}^{\mu,\widetilde{\mu}}( \mu(w) \cup \{f'\}) | = |\mu(w)|,$ and the proof of the claim is completed. Now, we have two subcases to consider:

\noindent \textbf{Subcase 2.1:   $\boldsymbol{f' \in C_{w}^{\mu,\widetilde{\mu}}( \mu(w) \cup \{f'\})}$}. 
As $| C_{w}^{\mu,\widetilde{\mu}}( \mu(w) \cup \{f'\}) | = |\mu(w)|,$ there is $f \in \mu(w)$ such that 
\begin{equation}\label{terminalnode1}
C^{\mu, \widetilde{\mu}}_{w}\left(\mu(w)\cup \{f'\}\right)=\left(\mu(w)\setminus \{f\}\right)\cup \{f'\}.
\end{equation}
 Furthermore, $f \notin \widetilde{\mu}(w).$ To see this, notice that if $f \in \widetilde{\mu}(w)=C_{w}^{\mu,\widetilde{\mu}}(F)$ then,  by substitutability, $f \in C_{w}^{\mu,\widetilde{\mu}}(\mu(w)\cup \{f'\}),$ contradicting \eqref{terminalnode1}. Therefore, $f \in \mu(w)\setminus \widetilde{\mu}(w)$ and \eqref{terminalnode1} imply that there is an arc from $(f',w)\in V_2$ to  $(w, f)\in V_1$. This contradicts that $(f',w)$ is a terminal node of $p.$ 

\noindent \textbf{Subcase 2.2:   $\boldsymbol{f' \notin C_{w}^{\mu,\widetilde{\mu}}( \mu(w) \cup \{f'\})}$}. First, assume that $f' \notin C_{w}( \mu(w) \cup \{f'\})$. Since $(f',w)$ is the terminal node of path $p$, there are $(w',f') \in V_1$ and an arc from $(w',f')$ to $(f',w).$ Also,  $f' \notin \mu(w),$ implying that $C_{w}( \mu(w) \cup \{f'\})= \mu(w)$. Thus, by Step 2 (b) of the reduction procedure, $\{f'\}$ is eliminated from $w$'s preference list. Thus, by Step 3 of the reduction procedure, all subsets of workers containing $w$ are eliminated from preference list of $f'$ as well. This contradicts that $(f',w)\in V_2.$ 
Second, assume that $f' \in C_{w}( \mu(w) \cup \{f'\})$. Since $f' \notin \mu(w)$ and $C_{w}( \mu(w) \cup \{f'\}) \neq \mu(w)$, then  $\{f'\}$ is not eliminated on Step 2 (b) of the reduction procedure. Moreover, Step 3 of the reduction procedure  neither eliminates $f'$ nor $w$ from each other's preference lists, because $(f',w) \in V_2$. Then, $C_{w}^{\mu,\widetilde{\mu}}( \mu(w) \cup \{f'\})=C_{w}( \mu(w) \cup \{f'\})$, implying that $f' \in C_{w}^{\mu,\widetilde{\mu}}( \mu(w) \cup \{f'\})$,  contradicting this subcase's hypothesis.

Therefore, by Cases 1 and 2, path $p$ has no terminal node so it is a  cycle in digraph $D^{\mu, \widetilde{\mu}}.$ As a consequence, a cycle  for $P^{\mu,\widetilde{\mu}}$ must also exist.
\end{proof}

\bigskip

\noindent\begin{proof}[Proof of Proposition \ref{ciclico es estable}]Let $\mu'$ be a cyclic matching under $P^{\mu,\widetilde{\mu}}$. Let $\sigma $ be the cycle associated with $\mu ^{\prime }.$ First, we prove that $\mu'$ is an individually rational matching under $P^{\mu,\widetilde{\mu}}.$ If $a\in F\cup W$ with $a\notin \sigma$, we have that  $\mu'(a)=\mu(a)$. Then, by the individual rationality of $\mu$ under $P^{\mu,\widetilde{\mu}}$ we have that $C_a^{\mu,\widetilde{\mu}}(\mu'(a))=\mu'(a).$ If $f\in \sigma$, there is $w'\in \sigma$ such that $\mu'(f)= C_f^{\mu,\widetilde{\mu}}(W \setminus \{w'\})$. Thus, $C_f^{\mu,\widetilde{\mu}}(\mu'(f))=C_f^{\mu,\widetilde{\mu}}(C_f^{\mu,\widetilde{\mu}}(W \setminus \{w'\}))=C_f^{\mu,\widetilde{\mu}}(W \setminus \{w'\})=\mu'(f).$ If $w \in\sigma$, there is $f\in \mu(w)$ and $f'\in \sigma$ such that $\mu'(w)=(\mu(w)\setminus\{f\})\cup \{f'\}.$ Then, $C_w^{\mu,\widetilde{\mu}}(\mu'(w))=C_w^{\mu,\widetilde{\mu}}((\mu(w)\setminus\{f\})\cup \{f'\})$. By definition of a cycle, $C_w^{\mu,\widetilde{\mu}}((\mu(w)\setminus\{f\})\cup \{f'\})=C_w^{\mu,\widetilde{\mu}}(C_w^{\mu,\widetilde{\mu}}(\mu(w)\cup \{f'\}))=C_w^{\mu,\widetilde{\mu}}(\mu(w)\cup \{f'\})=\mu'(w).$ Therefore, $\mu'$ is individually rational under $P^{\mu,\widetilde{\mu}}.$

Second, assume that there is a blocking pair $(f,w)$ for $\mu'$ under  $P^{\mu,\widetilde{\mu}}.$   We claim that both $f$ and $w$ belong to $\sigma.$ Furthermore, $w$ immediately precedes $f$ in cycle $\sigma.$ In order to see this, first assume that $f \notin \sigma.$ Thus,  by the definition of cyclic matching,  $\mu
^{\prime }(f)=\mu (f)$ and since,  by Remark \ref{remark de reduccion M-M} (i), $\mu ^{\prime }(f)$ is the most preferred subset of workers in $P_{f}^{\mu,\widetilde{\mu}}$, there is no $w'\notin \mu ^{\prime
}(f)$ such that $w' \in C _{f}^{\mu,\widetilde{\mu}}(\mu ^{\prime }(f)\cup \{w'\})$. When $w'=w,$ this contradicts that $(f,w)$ is a blocking pair for $\mu'$. Therefore, $f \in \sigma.$ 

 Also, as $(f,w)$ is a blocking pair for $\mu'$,  $w\in
C_{f}^{\mu ,\widetilde{\mu}}(\mu ^{\prime }(f)\cup \{w\}).$ By the definition of cycle, there is $w'$ such that $C_{f}^{\mu
,\widetilde{\mu}}(W\setminus \{w'\})=\mu'(f)$ and thus $w\in C_{f}^{\mu ,\widetilde{\mu}}(C_{f}^{\mu,\widetilde{\mu}}(W\setminus
\{w' \})\cup \{w\})$ which in turn,  by \eqref{propiedad 1}, becomes
\begin{equation}\label{ciclico es estable 1}
w\in C_{f}^{\mu,\widetilde{\mu}}((W\setminus \{w'\})\cup \{w\}).
\end{equation}
To see that $w$ immediately precedes $f$ in cycle $\sigma,$ i.e. $w = w',$ assume that $w\neq w'.$ Then, $%
w\in W\setminus \{w'\}$ and, therefore, \eqref{ciclico es estable 1} implies  $w\in C_{f}^{\mu
,\widetilde{\mu}}(W\setminus \{w'\})=\mu'(f).$ Thus, $w\in \mu ^{\prime }(f)$, which contradicts $(f,w)$ being a blocking pair for $\mu'$. Hence, $w=w'.$ This completes our claim. 

To finish our proof, notice that by definition of cyclic matching and the fact that  $w=w' \in \sigma,$ there is $f'$ such that 
\begin{equation}\label{ciclico es estable 2}
C_{w}^{\mu,\widetilde{\mu}}(\mu(w)\cup \{f'\})=(\mu(w) \setminus \{f\}) \cup \{f'\}=\mu'(w).
\end{equation}
Since $\mu(w)\cup \{f'\}=\mu'(w)\cup \{f\},$ using \eqref{ciclico es estable 2} and $f \notin \mu'(w)$ (that follows from $(f,w)$ being a blocking pair for $\mu'$) we have that $f \notin C_{w}^{\mu,\widetilde{\mu}}(\mu'(w)\cup \{f\}).$ But then  again we contradict that $(f,w)$ is a blocking pair for $\mu'.$ Hence,    $\mu ^{\prime }\in S\left( P^{\mu,\widetilde{\mu}}\right).$
\end{proof}

\bigskip

\noindent \begin{proof}[Proof of Proposition \ref{matching ciclicio arriba del peor}]Let $\mu, \mu' \in S(P)$ with $\mu \succ_F \mu'$. Consider the reduced preference profile $P^{\mu,\mu'}$. By Proposition \ref{existencia de ciclo}, there is a cycle $\sigma$ for $P^{\mu,\mu'}$. Let $\mu_{\sigma}$ be its corresponding cyclic matching under $P^{\mu,\mu'}$. By Proposition \ref{ciclico es estable}, $\mu_{\sigma}\in S(P^{\mu,\mu'})$ and, consequently, $\mu_{\sigma}\succeq_{F}\mu'$ by Lemma \ref{estable original sii estable en el reducido M-to-M}. Furthermore, $\mu \succeq_{F}\mu_{\sigma}$ follows straightforward from the fact that $\mu_{\sigma}\in S(P^{\mu})$ and that $\mu$ is the firm-optimal stable matching for $P^{\mu}$.
\end{proof}

\begin{lemma}\label{lema4}
Let $\mu,\widetilde{\mu}\in S(P)$ with $\mu\succeq_F\widetilde{\mu}$. If $\widetilde{\mu}$ is a cyclic matching under  $P^{\mu,\widetilde{\mu}}$, then $\widetilde{\mu}$ is a cyclic matching under  $P^{\mu}$. 
\end{lemma}
\begin{proof}
Let $\mu,\widetilde{\mu}\in S(P)$ with $\mu\succeq_F\widetilde{\mu}$ and let $\widetilde{\mu}$ be a cyclic matching under  $P^{\mu,\widetilde{\mu}}$. By Theorem \ref{estable original sii estable en el reducido M-to-M}, $\widetilde{\mu}\in S(P^{\mu})$. Let $\sigma=\{(w_{1},f_{1}),(w_{2},f_{2}),\ldots,(w_r ,f_{r})\}$ be the cycle associated with $\widetilde{\mu}$. We only need to prove that $\sigma$ is a cycle for $P^\mu$. First, notice that for each $(w_i,f_i)\in\sigma$, $w_i \in \mu(f_i)\setminus\widetilde{\mu}(f_i)$ implies that $w_i \in \mu(f_i)\setminus\mu_W(f_i).$ Otherwise,  $w_i \in \mu_W(f_i)$ and $w_i \in \mu(f_i)= C_{f_i}^{\mu} (\mu(f_i) \cup \widetilde{\mu}(f_i) \cup \mu_W(f_i))$ imply, by substitutability, that $w_i\in C_{f_i}^{\mu} (\widetilde{\mu}(f_i) \cup \mu_W(f_i))= \widetilde{\mu}(f_i)$, a contradiction. Second, by definition of cycle for $P^{\mu,\widetilde{\mu}}$ and Definition \ref{defino matching ciclico},  $\widetilde{\mu}(f_i)=C_{f_i}^{\mu,\widetilde{\mu}}(W \setminus \{w_i\})=(\mu(f_i) \setminus \{w_i\})\cup \{w_{i+1}\}.$ By Proposition \ref{ciclico es estable}, $\widetilde{\mu}$ is individually rational under  $P^{\mu,\widetilde{\mu}}$. By Lemma \ref{individual dentro de mu y mu tilde}, $\widetilde{\mu}$ is individually rational under  $P^{\mu}.$ Thus, $C_{f_i}^{\mu}(\widetilde{\mu}(f_i))=\widetilde{\mu}(f_i).$ Hence, 
\begin{equation}\label{ecu lema apendice 1}
C_{f_i}^{\mu}(\widetilde{\mu}(f_i))=(\mu(f_i) \setminus \{w_i\})\cup \{w_{i+1}\}.
\end{equation} 
Lastly, again by definition of cycle for $P^{\mu,\widetilde{\mu}}$, we have 
$$C^{\mu,\widetilde{\mu}}_{w_i}(\mu(w_i) \cup \{f_{i-1}\}) =\left(\mu(w_i) \setminus \{f_i\}\right)\cup \{f_{i-1}\}=\widetilde{\mu}(w_i).$$
Now, we prove that $C^{\mu,\widetilde{\mu}}_{w_i}(\mu(w_i) \cup \{f_{i-1}\})=C^{\mu}_{w_i}(\mu(w_i) \cup \{f_{i-1}\}).$ 
By the reduction procedure, we have that $C^{\mu,\widetilde{\mu}}_{w_i}(\mu(w_i) \cup \{f_{i-1}\})\subseteq C^{\mu}_{w_i}(\mu(w_i) \cup \{f_{i-1}\}).$
Now, assume that  $ C^{\mu}_{w_i}(\mu(w_i) \cup \{f_{i-1}\})\neq C^{\mu,\widetilde{\mu}}_{w_i}(\mu(w_i) \cup \{f_{i-1}\}).$ This implies that $ C^{\mu}_{w_i}(\mu(w_i) \cup \{f_{i-1}\})$  is eliminated in the reduction procedure to obtain $P^{\mu,\widetilde{\mu}}$. Since $\mu\in S(P^{\mu,\widetilde{\mu}})$, then the only possibility is that the firm selected by the reduction procedure to eliminate from $ C^{\mu}_{w_i}(\mu(w_i) \cup \{f_{i-1}\})$ be $f_{i-1}$. This contradicts that $ \widetilde{\mu} $ is individually rational under $ P^{\mu,\widetilde{\mu}} $, because $ f_{i-1}\in  \widetilde{\mu}(w_i).$
\end{proof}

\noindent \begin{proof}[Proof of Proposition \ref{todo estable es ciclico}]Let $\mu'\in S(P)\setminus \{\mu_F\}$ and consider the reduced preference profile $P^{\mu_F,\mu'}$. If $\mu'$ is a cyclic matching under $P^{\mu_F,\mu'}$, then by Lemma \ref{lema4} $\mu'$ is a cyclic matching under $P^{\mu_F}$ and the proof is complete. If not, by Proposition \ref{matching ciclicio arriba del peor} there is a cyclic matching under $P^{\mu_F,\mu'}$, say $\mu_1$, such that $\mu_1\succ_F \mu'$. By Lemma \ref{estable original sii estable en el reducido M-to-M} and Proposition \ref{ciclico es estable}, $\mu_1\in S(P)$, so we can consider the reduced preference profile $P^{\mu_1,\mu'}$. If $\mu'$ is a cyclic matching under $P^{\mu_1,\mu'}$, then by Lemma \ref{lema4} $\mu'$ is a cyclic matching under $P^{\mu_1}$, and the proof is complete. If not, continue this process until, by the finiteness of $S(P)$, there is $\mu_k\in S(P)$ such that $\mu'$ is a cyclic matching under $P^{\mu_k,\mu'}$ , then by Lemma \ref{lema4} $\mu'$ is a cyclic matching under $P^{\mu_k}$.
\end{proof}

\noindent \begin{proof}[Proof of Theorem \ref{Teorema final}]
Let $(F,W,P)$ be a matching market. First, notice that by Proposition \ref{existencia de ciclo}, for each reduced profile obtained in Step $t-1$, there is at least a cycle. Proposition \ref{ciclico es estable} and Theorem \ref{estable original sii estable en el reducido M-to-M} show that each cyclic matching obtained by the algorithm belongs to $S(P)$. To see that each stable matching is computed by the algorithm, assume that it is not the case for $\mu \in S(P)\setminus \{\mu_F\}.$ By Proposition  \ref{todo estable es ciclico}, there is another $\mu'\in S(P)$ such that $\mu$ is a cyclic matching under $P^{\mu'}$ (remember that, as $\mu$ is a cyclic matching under $P^{\mu'}$, $\mu' \succ_F \mu$).  Thus, $\mu'$ is not computed by the algorithm either (otherwise, if $\mu'$ is computed by the algorithm in Step $t$, $\mu$ necessarily is computed in Step $t+1$). Thus, again by Proposition  \ref{todo estable es ciclico}, there is another $\mu''\in S(P)$ such that $\mu'$ is a cyclic matching under $P^{\mu''}$ with $\mu'' \succ_F \mu'$ and $\mu''$ is not computed by the algorithm either. Continuing this line of reasoning, by the finiteness of the set $S(P)$, we eventually reach $\mu_F$ and conclude that the algorithm cannot compute it either, which is absurd. 
\end{proof}


\begin{thebibliography}{18}
\newcommand{\enquote}[1]{``#1''}
\expandafter\ifx\csname natexlab\endcsname\relax\def\natexlab#1{#1}\fi

\bibitem[\protect\citeauthoryear{Alkan}{Alkan}{2002}]{alkan2002class}
\textsc{Alkan, A.} (2002): \enquote{A class of multipartner matching markets
  with a strong lattice structure,} \emph{Economic Theory}, 19, 737--746.

\bibitem[\protect\citeauthoryear{Bansal, Agrawal, and Malhotra}{Bansal
  et~al.}{2007}]{bansal2007polynomial}
\textsc{Bansal, V., A.~Agrawal, and V.~Malhotra} (2007): \enquote{Polynomial
  time algorithm for an optimal stable assignment with multiple partners,}
  \emph{Theoretical Computer Science}, 379, 317--328.

\bibitem[\protect\citeauthoryear{Blair}{Blair}{1988}]{blair1988lattice}
\textsc{Blair, C.} (1988): \enquote{The lattice structure of the set of stable
  matchings with multiple partners,} \emph{Mathematics of Operations Research},
  13, 619--628.

\bibitem[\protect\citeauthoryear{Cheng, McDermid, and Suzuki}{Cheng
  et~al.}{2008}]{cheng2008unified}
\textsc{Cheng, C., E.~McDermid, and I.~Suzuki} (2008): \enquote{A unified
  approach to finding good stable matchings in the hospitals/residents
  setting,} \emph{Theoretical Computer Science}, 400, 84--99.

\bibitem[\protect\citeauthoryear{Deng, Panigrahi, and Waggoner}{Deng
  et~al.}{2017}]{deng2017complexity}
\textsc{Deng, Y., D.~Panigrahi, and B.~Waggoner} (2017): \enquote{The
  complexity of stable matchings under substitutable preferences,} in
  \emph{Thirty-First AAAI Conference on Artificial Intelligence}.

\bibitem[\protect\citeauthoryear{Dworczak}{Dworczak}{2021}]{dworczak2021deferred}
\textsc{Dworczak, P.} (2021): \enquote{Deferred acceptance with compensation
  chains,} \emph{Operations Research}, 69, 456--468.

\bibitem[\protect\citeauthoryear{Eirinakis, Magos, Mourtos, and
  Miliotis}{Eirinakis et~al.}{2012}]{eirinakis2012finding}
\textsc{Eirinakis, P., D.~Magos, I.~Mourtos, and P.~Miliotis} (2012):
  \enquote{Finding all stable pairs and solutions to the many-to-many stable
  matching problem,} \emph{INFORMS Journal on Computing}, 24, 245--259.

\bibitem[\protect\citeauthoryear{Gale and Shapley}{Gale and
  Shapley}{1962}]{gale1962college}
\textsc{Gale, D. and L.~Shapley} (1962): \enquote{College admissions and the
  stability of marriage,} \emph{The American Mathematical Monthly}, 69, 9--15.

\bibitem[\protect\citeauthoryear{Gusfield and Irving}{Gusfield and
  Irving}{1989}]{gusfield1989stable}
\textsc{Gusfield, D. and R.~W. Irving} (1989): \emph{The stable marriage
  problem: structure and algorithms}, MIT press.

\bibitem[\protect\citeauthoryear{Hatfield and Milgrom}{Hatfield and
  Milgrom}{2005}]{hatfield2005matching}
\textsc{Hatfield, J. and P.~Milgrom} (2005): \enquote{Matching with contracts,}
  \emph{American Economic Review}, 95, 913--935.

\bibitem[\protect\citeauthoryear{Irving and Leather}{Irving and
  Leather}{1986}]{irving1986complexity}
\textsc{Irving, R. and P.~Leather} (1986): \enquote{The complexity of counting
  stable marriages,} \emph{SIAM Journal on Computing}, 15, 655--667.

\bibitem[\protect\citeauthoryear{Kelso and Crawford}{Kelso and
  Crawford}{1982}]{kelso1982job}
\textsc{Kelso, A. and V.~Crawford} (1982): \enquote{Job matching, coalition
  formation, and gross substitutes,} \emph{Econometrica}, 50, 1483--1504.

\bibitem[\protect\citeauthoryear{Li}{Li}{2014}]{li2014new}
\textsc{Li, J.} (2014): \enquote{A new proof of the lattice structure of
  many-to-many pairwise-stable matchings,} \emph{Journal of the Operations
  Research Society of China}, 2, 369--377.

\bibitem[\protect\citeauthoryear{Mart{\'i}nez, Mass{\'o}, Neme, and
  Oviedo}{Mart{\'i}nez et~al.}{2004}]{martinez2004algorithm}
\textsc{Mart{\'i}nez, R., J.~Mass{\'o}, A.~Neme, and J.~Oviedo} (2004):
  \enquote{An algorithm to compute the full set of many-to-many stable
  matchings,} \emph{Mathematical Social Sciences}, 47, 187--210.

\bibitem[\protect\citeauthoryear{McVitie and Wilson}{McVitie and
  Wilson}{1971}]{mcvitie1971stable}
\textsc{McVitie, D. and L.~Wilson} (1971): \enquote{The stable marriage
  problem,} \emph{Communications of the ACM}, 14, 486--490.

\bibitem[\protect\citeauthoryear{Roth}{Roth}{1984}]{roth1984evolution}
\textsc{Roth, A.} (1984): \enquote{The evolution of the labor market for
  medical interns and residents: a case study in game theory,} \emph{Journal of
  Political Economy}, 92, 991--1016.

\bibitem[\protect\citeauthoryear{Roth}{Roth}{1985}]{roth1985college}
---\hspace{-.1pt}---\hspace{-.1pt}--- (1985): \enquote{The college admissions
  problem is not equivalent to the marriage problem,} \emph{Journal of Economic
  Theory}, 36, 277--288.

\bibitem[\protect\citeauthoryear{Roth and Sotomayor}{Roth and
  Sotomayor}{1990}]{roth1992two}
\textsc{Roth, A. and M.~Sotomayor} (1990): \emph{Two-Sided Matching: A Study in
  Game-Theoretic Modeling and Analysis}, Cambidge University Press, Cambridge.

\end{thebibliography}
\end{document}